\documentclass[a4paper]{article}
\usepackage{a4wide}
\usepackage{graphicx,xcolor}
\graphicspath{{./images/}}
\usepackage{caption}
\usepackage{subcaption}
\usepackage{empheq}
\usepackage{amsthm}
\usepackage{amssymb, latexsym, mathrsfs}
\usepackage{amsmath}
\usepackage{dsfont}
\usepackage{enumerate}
\usepackage{algorithm}
\usepackage{color}
\usepackage{transparent}
\usepackage{url}
\usepackage[affil-it]{authblk}
\usepackage{booktabs}

\usepackage[american,cuteinductors,smartlabels]{circuitikz}
\usepackage{tikz}
\usepackage{schemabloc}
\usetikzlibrary{shapes,arrows}

\theoremstyle{plain}
\newtheorem{thm}{Theorem}
\newtheorem{assum}{Assumption}

\newtheorem{rmk}{Remark}
\newtheorem{prop}{Proposition}
\newtheorem{lem}{Lemma}

\newtheorem{prob}{Problem}

\newcommand{\Rset}{\mathbb{R}}

\newcommand{\hd}{{\hat{d}}}

\newcommand{\hx}{{\hat{x}}}

\newcommand{\CC}{{\mathcal{C}}}
\newcommand{\DD}{{\mathcal{D}}}
\newcommand{\EE}{{\mathcal{E}}}
\newcommand{\FF}{{\mathcal{F}}}
\newcommand{\GG}{{\mathcal{G}}}

\newcommand{\LL}{{\mathcal{L}}}
\newcommand{\MM}{{\mathcal{M}}}
\newcommand{\NN}{{\mathcal{N}}}
\newcommand{\OO}{{\mathcal{O}}}
\newcommand{\PP}{{\mathcal{P}}}
\newcommand{\QQ}{{\mathcal{Q}}}

\newcommand{\VV}{{\mathcal{V}}}

\newcommand{\diag}{{\mbox{diag}}}                  
\newcommand{\norme}[2]{||{#1}||_{#2}}            

\newcommand{\mbf}[1]{\mathbf{#1}}                  

\newcommand{\subss}[2]{{#1}_{[#2]}}

\begin{document}
	\title{\LARGE \bf Voltage stabilization in DC microgrids: an
          approach based on line-independent plug-and-play controllers}
          \author[1]{Michele Tucci%
       \thanks{Electronic address:
         \texttt{michele.tucci02@universitadipavia.it}; Corresponding author}}
        \author[2]{Stefano Riverso%
       \thanks{Electronic address: \texttt{riverss@utrc.utc.com}}} 
     \author[3]{Giancarlo Ferrari-Trecate%
       \thanks{Electronic address: \texttt{giancarlo.ferraritrecate@epfl.ch}} }

     \affil[1]{Dipartimento di Ingegneria Industriale e
       dell'Informazione\\Universit\`a degli Studi di Pavia}
     \affil[2]{United Technologies Research Center Ireland}  
     \affil[3]{Automatic Control Laboratory, \'Ecole Polytechnique F\'ed\'erale de Lausanne (EPFL), 1015 Lausanne, Switzerland.}
     \date{\textbf{Technical Report}\\ December, 2016}

     \maketitle
     \begin{abstract}
We consider the problem of stabilizing voltages in  DC
microGrids (mGs) given by the interconnection of Distributed Generation
Units (DGUs), power lines and loads. 
We propose a
decentralized control architecture where the primary controller of each DGU can be
designed in a Plug-and-Play (PnP) fashion, allowing the seamless addition of new DGUs.
Differently from several other approaches to primary control, local design is independent of 
the parameters of power lines. Moreover, differently from the PnP control scheme in \cite{Tucci2015a}, the plug-in of a DGU does not require to update controllers of neighboring DGUs. Local control design is cast into a  
Linear Matrix Inequality (LMI) problem that, if unfeasible, allows one to deny plug-in requests that might be dangerous for mG stability. 
The proof of closed-loop stability of voltages 
exploits structured Lyapunov functions, the
LaSalle invariance theorem and properties of graph Laplacians. 
Theoretical results are backed up by
simulations in PSCAD.
     \end{abstract}

\newpage
       \section{Introduction}
\label{sec:intro}
DC mGs are networks combining energy sources, interfaced through
converters, and loads connected through power lines. As DC mGs can be coupled to the main grid through AC-DC converters only, they can be thought as always operating in islanded mode. This is due to the fact that converters have finite power rating, which can limit substantially the power transfer.

Motivated by advances in power electronics, batteries and renewable DC
sources, DC mGs find nowadays applications in
various field such as high-efficiency households, electric vehicles,
hybrid energy storage systems, data centers, avionics and marine
systems \cite{dragicevic2016dc}.

In order to operate DC mGs in
a safe and reliable way, several challenges must be addressed. A central one is to
provide voltage stability \cite{dragicevic2015dc,7172122,zhao2015distributed,zonetti2014globally,Tucci2015a,Tucci2015dc}. This property is
fundamental and should be guaranteed even in presence of unreliable communication channels among local controllers or DGUs. This promoted the study of decentralized control
architectures at the primary level of the mG control hierarchy \cite{guerrero2011hierarchical}. The most popular solutions are based on droop controllers for
each DGU, built on top inner voltage and current loops
\cite{guerrero2011hierarchical}. Stability of droop control has been
however shown either for specific mG topologies \cite{shafiee2014hierarchical}
or for general topologies, but relying on networked secondary
regulators \cite{zhao2015distributed}. 

An alternative class of decentralized primary controllers,
termed PnP according to the terminology used in \cite{Riverso2013c,Riverso2014a} and \cite{7040312}, has been proposed in \cite{Tucci2015a,Tucci2015dc}. 
The main feature of the PnP approach is to allow the addition
and removal of DGUs independently of the mG size and without retuning
all local controllers. 

In particular, in \cite{Tucci2015a,Tucci2015dc}, the computation of a local
controller requires to solve an LMI problem based only on information about the
corresponding DGU and power lines connecting neighboring DGUs. This has two consequences. First, if a DGU 
is plugged-in, its neighbors must update their
controllers through LMIs, as they will be connected to new lines. 
Second, in order to preserve stability of the mG, the plug-in of the DGU must be 
denied if one of the LMI problems is infeasible. 
Another feature of PnP control is that LMIs give, as a byproduct, 
local structured Lyapunov functions
that can be used for certifying asymptotic stability of the whole mG.

The dependence of a local controller on power line parameters can be critical when they are not precisely known. This problem has been considered in 
\cite{Cezar_Rajagopal_Zhang_2015} and \cite{Josep_Hamzeh_Ghafouri_Member_Karimi_Sheshyekani_Guerrero_2016} where line-independent controllers have been proposed. However, in 
\cite{Cezar_Rajagopal_Zhang_2015} and \cite{Josep_Hamzeh_Ghafouri_Member_Karimi_Sheshyekani_Guerrero_2016} voltage stability has been studied only for mGs with at most four DGUs.

In this paper, we propose a variant of the PnP design algorithm in
\cite{Tucci2015a,Tucci2015dc}. The main novelty is that the
design of local controllers is line-independent and the only global quantity used in the synthesis algorithm is a scalar parameter.
%
As a consequence, even if a DGU
wants to plug -in or -out, its neighbors do not have to update
local controllers. This considerably simplifies the plug-in protocol
described in \cite{Tucci2015a,Tucci2015dc}, as switching of local controllers is avoided and
bumpless control architectures (for avoiding abrupt changes in the
control variables) are unnecessary. Moreover, the new design procedure better complies
with privacy requirements of energy markets where DGUs can
have different owners. Indeed, the addition of new DGUs does not
require other stakeholders to disclose models of their own DGUs or change their operation.
For the new design procedure, the structure of each local controller is identical to the
one proposed in \cite{Tucci2015a,Tucci2015dc}. Also LMI problems associated to control design
are similar to those in \cite{Tucci2015a,Tucci2015dc}. However, the proof of asymptotic stability of the
closed-loop system is substantially different and more involved. In
particular, it is based on the fact that, under Quasi Stationary Line (QSL) approximations,
electrical coupling among DGUs can be described by graph
Laplacians \cite{dorfler2013kron}. This feature, together with the use of structured Lyapunov
functions and the LaSalle Invariance
principle \cite{khalil2001nonlinear}, allows us to derive the desired result.

The paper is structured as follows. Results from \cite{Tucci2015a,Tucci2015dc} about the DGU model and the structure of
PnP controllers are summarized in Sections~\ref{sec:Model}  and
\ref{sec:aug_sys}. The new PnP approach to the design of local
controllers is described in Sections~\ref{sec:pnp_design}, \ref{subsec:LMI} and \ref{sec:PnP}, along with
the stability analysis of the closed-loop system. 
Simulations in PSCAD
using a 5-DGU mG and illustrating the plug-in of an additional DGU are
described in Section~\ref{sec:scenario_2}. 

\textbf{Notation.} We use $P>0$ (resp. $P\geq 0$) for indicating the
real symmetric matrix $P$ is positive-definite
(resp. positive-semidefinite). Let
	$A\in\mathbb{R}^{n\times m}$ be a
	matrix inducing the linear map $A:\mathbb{R}^m\rightarrow \mathbb{R}^n$. The \textit{image} and the \textit{nullspace} (or
	\textit{kernel}) of $A$ are indicated with $\text{Im}(A)$ and
	$\text{Ker}(A)$, respectively. The
	symbol $\oplus$ refers to the sum of subspaces that are orthogonal (also called \textit{orthogonal direct sum}).
\section{DC Microgrid model}
          \label{sec:Model}      
\subsection{DGU electrical model}
\label{sec:el_model}
In this section, we describe the electrical model of a DC mG considered in 
\cite{Tucci2015a,Tucci2015dc}.
The electrical scheme of the $i$-th DGU represented within
the dashed frame in Figure \ref{fig:ctrl_part}.

The DC voltage source represents a generic renewable
resource\footnote{This approximation is reasonable
	since renewable power fluctuations
	take place at a slow timescale, compared to the one we are interested in for stability analysis. Moreover, renewables are
	usually equipped with storage units damping stochastic
	fluctuations.} and a Buck converter is commanded in order to supply a local DC
load connected to the Point of Common Coupling (PCC) through a series
$RL$ filter. We also assume that loads are unknown and act
as current disturbances ($I_L$) \cite{Tucci2015a,Tucci2015dc,Babazadeh2013}. 
Let us consider an mG composed of $N$ DGUs and let $\DD=\{1,\dots,N\}$. We call
two DGUs neighbors if there is a
power line connecting them and denote with
$\NN_i\subset\DD$ the subset of neighbors of DGU
$i$. The neighboring relation is symmetric:
$j\in\NN_i$ implies
$i\in\NN_j$. Furthermore, let $\mathcal
E=\{(i,j):$ $i\in\DD,$ $j\in\NN_i\}$ collect unordered pairs of
indices associated to lines. 
The topology of the mG is then described by the
undirected graph $\GG_{el}$ with nodes $\DD$ and
edges $\EE$.

From Figure \ref{fig:ctrl_part}, by applying Kirchoff's voltage and current laws, and exploiting QSL approximation
of power lines \cite{Venkatasubramanian1995,Tucci2015a,Tucci2015dc}, we obtain the following model of DGU
$i$
\begin{equation}
	\label{eq:newDGU}
	\text{DGU}~i:\hspace{-4mm}\quad\left\lbrace
	\begin{aligned}
		\frac{dV_{i}}{dt} &= \frac{1}{C_{ti}}I_{ti}+\sum\limits_{j\in\NN_i}\left(\frac{V_j}{C_{ti} R_{ij}}-\frac{V_i}{C_{ti}R_{ij}}\right)-\frac{1}{C_{ti}}I_{Li}\\
		\frac{dI_{ti}}{dt} &= -\frac{1}{L_{ti}}V_{i}-\frac{R_{ti}}{L_{ti}}I_{ti}+\frac{1}{L_{ti}}V_{ti}\\
	\end{aligned}
	\right.
\end{equation}
where variables $V_i$, $I_{ti}$, are the $i$-th PCC voltage and
filter current, respectively, $V_{ti}$ represents the command to the
Buck converter, and $R_{ti}$, $L_{ti}$ and $C_{ti}$ the converter electrical
parameters. Moreover, $V_{j}$ is the voltage at the
PCC of each neighboring DGU $j\in\NN_i$ and $R_{ij}$ is the resistance
of the power DC line connecting DGUs $i$ and $j$. 
\begin{figure}
	\centering
	\includegraphics[scale=0.5]{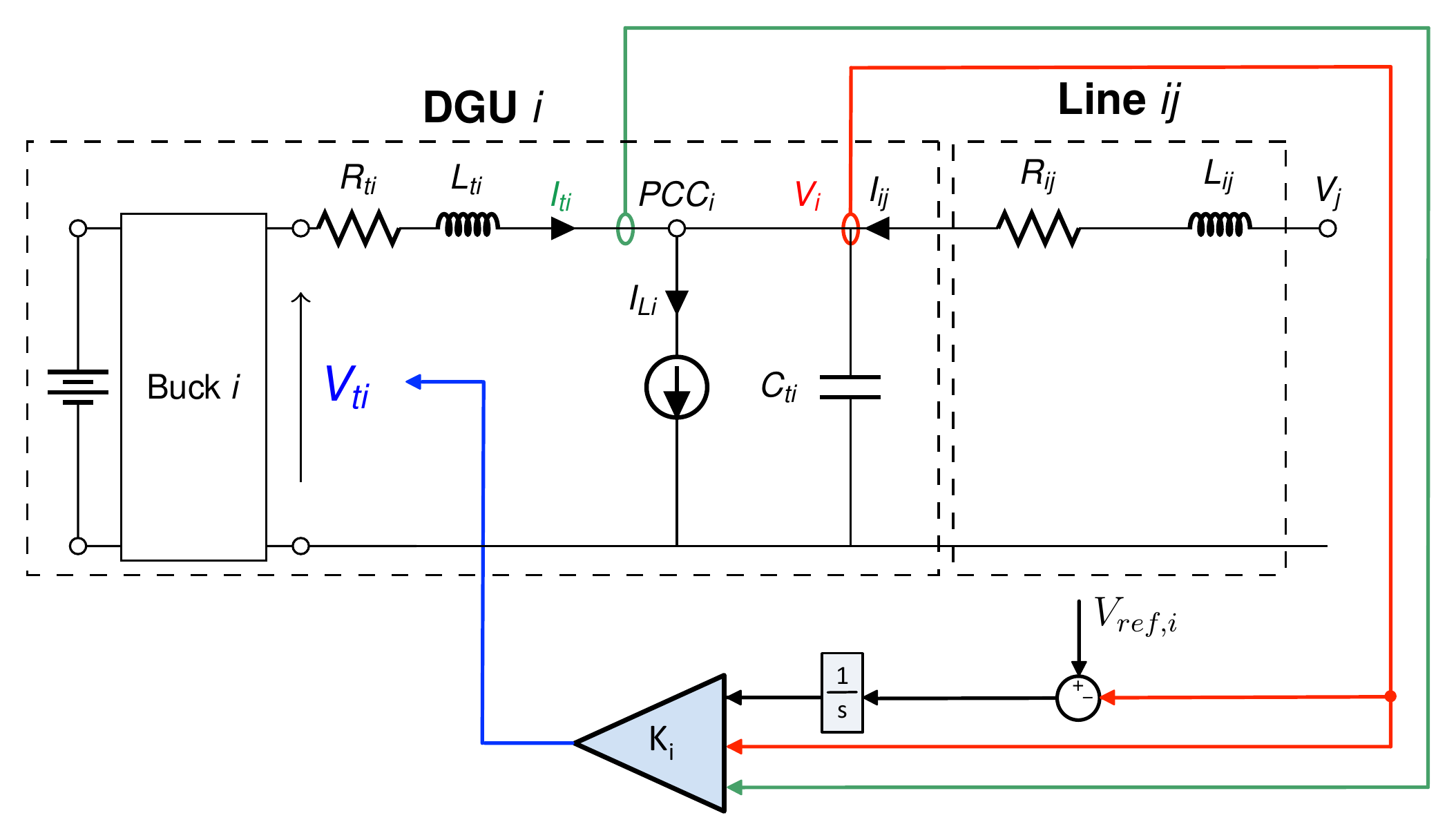}
	\caption{Electrical scheme of DGU $i$ (in the dashed frame) and local PnP voltage controller.}
	\label{fig:ctrl_part}
\end{figure}
\begin{rmk}
	Model \eqref{eq:newDGU} hinges on three main assumptions. First, Buck converter dynamics, that are inherently switching, have been averaged over time. This is however a mild approximation for modern converters that can operate at very high frequencies. Second, QSL approximations have been used. They amount to assume that power lines are mainly resistive and they have been justified in terms of singular perturbation theory \cite{Tucci2015a}. Third, loads are connected at the PCC of each DGU. It has been  shown that general interconnections of loads and DGUs can always be mapped into this topology using a network reduction method known as Kron reduction \cite{dorfler2013kron}.
\end{rmk}
\subsection{State-space model of the mG}
\label{sec:State-space model of the mG}
Dynamics \eqref{eq:newDGU} can be written in terms of
state-space variables as follows 
\begin{equation*}
	\text{ $\subss{\Sigma}{i}^{DGU}:$}\left\lbrace
	\begin{aligned}
		\subss{\dot{x}}{i}(t) &= A_{ii}\subss{x}{i}(t) +
		B_{i}\subss{u}{i}(t)+M_{i}\subss{d}{i}(t)+\subss\xi i(t)\\
		\subss{z}{i}(t)       &= H_{i}\subss{x}{i}(t)\\
	\end{aligned}
	\right.
\end{equation*}
where $\subss{x}{i}=[V_{i},I_{ti}]^T$ is the state,
$\subss{u}{i} = V_{ti}$ the control input,
$\subss{d}{i} = I_{Li}$ the exogenous input and
$\subss{z}{i} = V_{i}$ the controlled variable of the
system. 
The
term $\subss\xi i=\sum_{j\in\NN_i}A_{ij}(\subss x j-\subss x i)$ accounts
for the coupling with each DGU $j\in\NN_i$. 
This model is identical to the one provided in \cite{Tucci2015a,Tucci2015dc}, except that all
coupling terms have been embedded in variables $\subss{\xi} i$.
The matrices of
$\subss{\Sigma}{i}^{DGU}$ are obtained from
\eqref{eq:newDGU} as:

\begin{equation*}
	\renewcommand\arraystretch{1.5}
	A_{ii}=\begin{bmatrix}
		0 & \frac{1}{C_{ti}} \\
		-\frac{1}{L_{ti}} & -\frac{R_{ti}}{L_{ti}} \\
	\end{bmatrix},   \hspace{3mm} A_{ij}=
	\begin{bmatrix}
		\frac{1}{R_{ij}C_{ti}}  & 0 \\
		0 & 0 
	\end{bmatrix},
\end{equation*}
\begin{equation*}
	B_{i}=\begin{bmatrix}
		0 \\
		\frac{1}{L_{ti}}
	\end{bmatrix},
	\qquad
	M_{i}=\begin{bmatrix}
		-\frac{1}{C_{ti}} \\
		0 \\
	\end{bmatrix},
	\qquad
	H_{i}=\begin{bmatrix}
		1 & 0 
	\end{bmatrix}.
\end{equation*}
The overall mG model is given by
\begin{equation}
	\begin{aligned}
		\label{eq:stdformA}
		\mbf{\dot{x}}(t) &= \mbf{Ax}(t) + \mbf{Bu}(t)+ \mbf{Md}(t)\\
		\mbf{z}(t)       &= \mbf{Hx}(t)
	\end{aligned}
\end{equation}
where $\mbf x = (\subss x 1,\ldots,\subss x
N)\in\Rset^{2N}$, $\mbf u = (\subss u 1,\ldots,\subss u
N)\in\Rset^{N}$, $\mbf d = (\subss d 1,\ldots,\subss d
N)\in\Rset^{N}$, 
$\mbf z = (\subss z 1,\ldots,\subss z
N)\in\Rset^{N}$. Matrices $\mbf{A}$, $\mbf{B}$, $\mbf
M$ and $\mbf H$ are reported in Appendix
A.2 and A.3 of \cite{Tucci2015dc}. 

     \section{Design of stabilizing voltage controllers}
         \label{sec:PnPctrl}
        \subsection{Structure of local controllers}
        \label{sec:aug_sys}
        Let $\mbf{z_{ref}}(t)$ be the desired
        reference trajectory for the output $\mbf{z}(t)$. As in
        \cite{Tucci2015a,Tucci2015dc}, in
        order to track constant references $\mbf{z_{ref}}(t)=\mbf{\bar
        	z_{ref}}$, when
        $\mbf{d}(t)=\mbf{\bar d}$ is constant as well, we augment the mG
        model with integrators \cite{Skogestad1996}. A
        necessary condition for steering to zero the
        error $\mbf{e}(t)=\mbf{z_{ref}}(t)-\mbf{z}(t)$ as $t\rightarrow\infty$, is that, for arbitrary
        $\mbf{\bar d}$ and $\mbf{\bar
        	z_{ref}}$, there are equilibrium states and inputs $\mbf{\bar
        	x}$ and $\mbf{\bar u}$ verifying \eqref{eq:stdformA}.

        The existence of these equilibria can be shown following the proof of
        Proposition 1 in \cite{Tucci2015a}.
        
        The dynamics of the integrators is (see Figure
        \ref{fig:ctrl_part}, where $z_{ref_{[i]}}=V_{ref,i}$)
        \begin{equation}
        	\begin{aligned}
        		\subss{\dot{v}}{i}(t) = \subss{e}{i}(t) &= \subss{z_{ref}}{i}(t)-\subss{z}{i}(t) \\
        		&= \subss{z_{ref}}{i}(t)-H_{i}\subss{x}{i}(t),
        	\end{aligned}
        	\label{eq:intdynamics}
        \end{equation}
        and hence, the augmented DGU model is
        \begin{equation}
        	\label{eq:modelDGUgen-aug}
        	\subss{\hat{\Sigma}}{i}^{DGU} :
        	\left\lbrace
        	\begin{aligned}
        		\subss{\dot{\hat{x}}}{i}(t) &= \hat{A}_{ii}\subss{\hat{x}}{i}(t) + \hat{B}_{i}\subss{u}{i}(t)+\hat{M}_{i}\subss{\hat{d}}{i}(t)+\subss{\hat\xi}i(t)\\
        		\subss{z}{i}(t)       &= \hat{H}_{i}\subss{\hat{x}}{i}(t)
        	\end{aligned}
        	\right.
        \end{equation}
        where $\subss{\hat{x}}{i}=[\subss{x^T}
        i,v_{i,}]^T\in\Rset^3$ is the state,
        $\subss{\hat{d}}{i}=[\subss{d}{i},\subss{z_{ref}}{i}]^T\in\Rset^2$
        collects the exogenous signals and
        $\subss{\hat\xi}i=\sum_{j\in\NN_i}\hat{A}_{ij}(\subss{\hat{x}}{j}-\subss{\hat{x}}{i})$. Matrices
        in \eqref{eq:modelDGUgen-aug} are defined as follows
        \begin{equation*}
        	\begin{aligned}
        		\hat{A}_{ii} &=\begin{bmatrix}
        			A_{ii} & 0\\
        			-H_{i}C_{i} & 0
        		\end{bmatrix},
        		\hat{A}_{ij}=\begin{bmatrix}
        			A_{ij} &0\\
        			0&0
        		\end{bmatrix},
        		\hat{B}_{i}=\begin{bmatrix}
        			B_{i}\\
        			0
        		\end{bmatrix},
        		\\
        		\hat{M}_{i}&=\begin{bmatrix}
        			M_{i} & 0 \\
        			0 & 1
        		\end{bmatrix},\hspace{5mm}
        		\hat{H}_{i}=\begin{bmatrix}
        			H_{i} & 0
        		\end{bmatrix}.
        	\end{aligned}
        \end{equation*}   
        As in Proposition 2 of \cite{Tucci2015a}, one can show that
        positivity of electrical parameters guarantees that the
        pair $(\hat{A}_{ii},\hat{B}_i)$ is controllable. Hence,
        system \eqref{eq:modelDGUgen-aug} can be stabilized.
        
        The overall augmented system is obtained from \eqref{eq:modelDGUgen-aug} as
        \begin{equation}
        	\label{eq:sysaugoverall_1}
        	\left\lbrace
        	\begin{aligned}
        		\mbf{\dot{\hat{x}}}(t) &= \mbf{\hat{A}\hat{x}}(t) + \mbf{\hat{B}u}(t)+ \mbf{\hat{M}\hat{d}}(t)\\
        		\mbf{z}(t)       &= \mbf{\hat{H}\hat{x}}(t)
        	\end{aligned}
        	\right.
        \end{equation}
        where $\mbf{\hat{x}}$ and $\mbf{\hat{d}}$ collect variables $\subss{\hat{x}}{i}$ and $\subss{\hat{d}}{i}$ respectively, and matrices $\mbf{\hat{A}}, \mbf{\hat{B}}, \mbf{\hat{M}}$ and $\mbf{\hat{H}}$ are obtained from systems \eqref{eq:modelDGUgen-aug}. 
        
        Now we equip each DGU $\subss{\hat{\Sigma}}{i}^{DGU}$ with the following state-feedback controller
        \begin{equation}
        	\label{eq:ctrldec}
        	\subss{\CC}{i}:\qquad \subss{u}{i}(t)=K_{i}\subss{\hat{x}}{i}(t)
        \end{equation}
        where $K_{i}=\left[k_{1,i}\text{ }k_{2,i}\text{ }k_{3,i}\right]\in\Rset^{1\times3}$. It turns out that, together with the integral action \eqref{eq:intdynamics}, controllers $\subss{\CC}{i}$, define a multivariable PI regulator, see Figure~\ref{fig:ctrl_part}. In particular, the overall control architecture is
        decentralized since the computation of
        $\subss{u}{i}$ requires the state of
        $\subss{\hat{\Sigma}}{i}^{DGU}$ only. It is important to highlight that, in general, decentralized design of local regulators can fail to guarantee voltage stability of the whole mG, if couplings among DGUs are neglected during the design phase (see Appendix A).	In the sequel, we show how structured Lyapunov functions can be used to ensure asymptotic stability of the whole mG, when DGUs are equipped with controllers \eqref{eq:ctrldec}.
        
	  \subsection{Conditions for stability of the closed-loop mG}
           \label{sec:pnp_design}
        In absence of coupling terms $\subss{\hat\xi}i(t)$, one
        would like to guarantee asymptotic stability of the nominal closed-loop subsystem 
        \begin{equation}
        		\label{eq:modelDGUgen-aug-closed}
        		\subss{\dot{\hat{x}}}{i}(t) = \underbrace{(\hat{A}_{ii}+\hat{B}_{i}K_i)}_{F_{i}}\subss{\hat{x}}{i}(t)+\hat{M}_{i}\subss{\hd}{i}(t).\\
        \end{equation}
        By direct calculation, one can show that $F_{i}$ has the following structure
        \begin{equation}
		\begin{aligned}
        			 \label{eq:Fi}
        			F_{i}&=\left[ \begin{array}{c|cc}
        				0 & f_{12,i} & 0 \\
        				\hline
        				f_{21,i} & f_{22,i} & f_{23,i}\\
        				f_{13,i} & 0 & 0
        			\end{array}\right]=\\
        			&=\left[
        			\renewcommand\arraystretch{1.6}
        			\begin{array}{c|cc}
        				0 & \frac{1}{C_t} & 0 \\
        				\hline
        				\frac{(k_{i,1}-1)}{L_{ti}} & \frac{(k_{i,2}-R_{ti})}{L_{ti}} & \frac{k_{i,3}}{L_{ti}}\\
        				-1 & 0 & 0
        			\end{array}\right]=\left[
        			\renewcommand\arraystretch{1.5} \begin{array}{c|cc}
        				0 & \FF_{12,i} \\
        				\hline
        				\FF_{21,i}& \FF_{22,i}
        			\end{array}\right].
        	\end{aligned}
        \end{equation}
        From Lyapunov theory, asymptotic stability of \eqref{eq:modelDGUgen-aug-closed} is equivalent to the existence of a Lyapunov function $\VV_i(\subss \hx i)=\subss \hx i ^T P_i \subss \hx i$ where  $P_{i}\in\Rset^{3\times3}$, $P = P^T>0$ and 
        \begin{equation}
        	\label{eq:Lyapeqnith}
        	Q_i = F_{i}^T
        	P_{i}+P_{i}F_{i}
        \end{equation}
        is negative definite. 
        In presence of nonzero coupling terms, we will show that asymptotic stability can be achieved under two additional conditions.
        The first one is the use of the following separable Lyapunov
        functions $$\VV_i(\hat{x}_{[i]})=
        \eta_i\hat{x}_{[i],1}^2 + [\hat{x}_{[i],2}\text{
        }\hat{x}_{[i],3}]\PP_{22,i}[\hat{x}_{[i],2}\text{ }\hat{x}_{[i],3}]^T,$$
        where
        \begin{equation}
        	\label{eqn:P22}
        	\PP_{22,i}=
        	\left[ \begin{array}{cc}
        		p_{22,i} & p_{23,i}  \\
        		p_{23,i} & p_{33,i}
        	\end{array}\right] .
        \end{equation}
        This requirement is summarized in the next assumption.
        \begin{assum} 
        	\label{ass:ctrl}
        	Gains $K_{i}$, $i\in\DD$
        	are designed such that, in \eqref{eq:Lyapeqnith}, 
        	the positive definite matrix $P_i$ has the structure
        	\begin{equation}
        		\label{eq:pstruct}
        		P_{i}=\left[
        		\renewcommand\arraystretch{1.5}
        		\begin{array}{c|c}
        			\eta_i & \mbf{0}_{1\times 2} \\
        			\hline
        			\mbf{0}_{2\times 1}  & \PP_{22,i}\\
        		\end{array}\right],
        	\end{equation}
        	where the entries of $\PP_{22,i}$ are arbitrary and $\eta_i>0$ is a local parameter.
        \end{assum}
        The second condition concerns the values of parameters $\eta_i$. 
        \begin{assum}
        	\label{ass:equal_ratio}
        	Given a constant $\bar\sigma>0$ (a parameter common to
        	all DGUs), parameters $\eta_i$ in \eqref{eq:pstruct} are given by
        	\begin{equation}
        		\label{eq:equal_ratio}
        		\eta_i = \bar\sigma C_{ti}
        		\hspace{7mm} \forall i\in\DD.
        	\end{equation}
        \end{assum}
        
        The next result shows that, under Assumption~\ref{ass:ctrl}, Lyapunov theory certifies, at most,  marginal
        stability of \eqref{eq:modelDGUgen-aug-closed}.
\begin{prop}
\label{pr:prop_1}
Under Assumption \ref{ass:ctrl}, the matrix $Q_i$ cannot be negative
definite. Moreover, if 
\begin{equation}
\label{eq:Qi_semidef}
Q_i \leq 0,
\end{equation} 
then $Q_i$ has the following structure:
\begin{equation}
Q_i= \label{eq:qstruct}
                     \left[ \begin{array}{c|cc}
                         0 & 0 & 0 \\
                         \hline
                         0 & q_{22,i} & q_{23,i}\\
                         0 & q_{23,i} & q_{33,i}\\
                         \end{array}\right] = \left[ \begin{array}{c|cc}
                         0 & \mbf{0}_{1\times 2}  \\
                         \hline
                         \mbf{0}_{2\times 1} & Q_{22,i}
                         \end{array}\right] .
\end{equation}
\end{prop}
\begin{proof}
By direct computation, from \eqref{eq:Fi} and \eqref{eq:pstruct} one has $q_{11,i}=0$, showing that $Q_i$ cannot
be negative definite, as its first minor is not positive. Moreover, if a negative semidefinite matrix has
a zero element on its diagonal, then the corresponding row and column
have zero entries. This basic property can be shown as follows. If $\QQ\in\Rset^{n\times n}$ is symmetric and negative semidefinite, then $x^T\QQ x\leq 0$, $\forall x\in\Rset^n$. Partitioning $\QQ$ and $x$ as
	$$
	\label{eq:Q_partitioned}
	\renewcommand\arraystretch{2}
	\QQ=\left[ \begin{array}{c|cc}
	\QQ_{11} & \widetilde{\QQ}^T \\
	\hline
	\widetilde{\QQ} & \widehat{\QQ}
	\end{array}\right]
	\hspace{5mm}x=\left[ \begin{array}{c}
	x_{11}\\
	\tilde{x}
	\end{array}\right],$$
	one obtains
	$$ x^T\QQ x = x_{11}^2\QQ_{11} + 2 x_{11}\widetilde{\QQ}^T\tilde{x}+\tilde{x}^T\widehat{\QQ}\tilde{x}.$$
	Without loss of generality, assume $\QQ_{11}= 0$. Then, for any $x$ with $\tilde x = 0$ and $x_{11}\neq 0$, one has $x^T\QQ x = 0$, i.e. $x$ is a maximizer of $x^T\QQ x.$ Consequently, it must hold $\frac{\mathrm{d}}{\mathrm{dt}}\left(x^T\QQ x\right)= 0$, i.e. $2\QQ x = 0$, yielding
	\begin{equation*}
	\left\lbrace
	\begin{aligned}
	\QQ_{11}x_{11}+\widetilde{\QQ}^T\tilde{x} &=0\\
	\widetilde{\QQ}x_{11}+\widehat{\QQ}\tilde{x} &=0\\
	\end{aligned}
	\right.
	\end{equation*}
	Using $\QQ_{11}=0$ and $\tilde{x} = 0$, the previous linear system reduces to $\widetilde{\QQ}x_{11}=0$, that implies $\widetilde{\QQ}=0$.
Concluding, \eqref{eq:Qi_semidef} implies \eqref{eq:qstruct}. 
\end{proof}
           
                    Consider now the overall closed-loop 
            mG
            model 
            \begin{equation}
            	\label{eq:sysaugoverallclosed}
            	\left\lbrace
            	\begin{aligned}
            		\mbf{\dot{\hat{x}}}(t) &= (\mbf{\hat{A}+\hat{B}K})\mbf{\hat{x}}(t)+ \mbf{\hat{M}\hd}(t)\\
            		\mbf{z}(t)       &= \mbf{\hat{H}\hat{x}}(t)
            	\end{aligned}
            	\right.
            \end{equation}
            obtained by combining \eqref{eq:sysaugoverall_1} and \eqref{eq:ctrldec}, with
            $\mbf{K}=\diag(K_{1},\dots,K_{N})$. Consider also the collective Lyapunov function
            \begin{equation}
            	\label{eq_coll_lyap}
            	\VV(\mbf{\hx})=\sum_{i=1}^N\VV_i(\hat{x}_{[i]})=\mbf{\hx}^T\mbf{P}\mbf{\hx}
            \end{equation}
            where $\mbf{P}=\diag(P_{1},\dots,P_{N})$. One has $\dot \VV(\mbf{\hx})=\mbf{\hx}^T\mbf{Q}\mbf{\hx}$ where
            \begin{equation}
            	\nonumber
            	\mbf{Q} = (\mbf{\hat{A}}+\mbf{\hat{B}K})^T
            	\mbf{P}+\mbf{P}(\mbf{\hat{A}}+\mbf{\hat{B}K}).
            \end{equation}
            A consequence of Proposition \ref{pr:prop_1} is that, under Assumption \ref{ass:ctrl}, the matrix  $\mbf{Q}$
            cannot be negative definite. At most, one has
            \begin{equation}
            	\label{eq:Lyapeqnoverall}
            	\mbf{Q} \leq 0.
            \end{equation}
            Moreover, even if \eqref{eq:Qi_semidef} holds for all $i\in\DD$, the
            inequality \eqref{eq:Lyapeqnoverall} might be violated because of the nonzero
            coupling terms $\hat{A}_{ij}$ in matrix $\mbf{\hat A}$.
            The next result shows that this cannot happen under
            Assumption \ref{ass:equal_ratio}.
            %
            \begin{prop}
            	\label{pr:semidefinite_abc}
            	Under Assumptions \ref{ass:ctrl} and \ref{ass:equal_ratio}, if gains $K_i$ are computed in order to fulfill \eqref{eq:Qi_semidef} for all $i\in\DD$, then \eqref{eq:Lyapeqnoverall} holds.
            \end{prop}
            \begin{proof}
            	Consider the following decomposition of matrix $\mbf{\hat
            		A}$
            	\begin{equation}
            		\label{eq:decomposition}
            		\mbf{\hat A} = \mbf{\hat A_D}+\mbf{\hat A_{\Xi}} + \mbf{\hat A_C},
            	\end{equation}
            	where $\mbf{\hat{A}_{D}}=\diag(\hat{A}_{ii},\dots,\hat{A}_{NN})$
            	collects the local dynamics only, while $\mbf{\hat{A}_{\Xi}}=\diag(\hat{A}_{\xi 1},\dots,\hat{A}_{\xi
            		N})$ with 
            	\begin{equation*}
            		\renewcommand\arraystretch{1.5}
            		\hat{A}_{\xi i}=\begin{bmatrix}
            			-\sum\limits_{j \in \NN_i}\frac{1}{R_{ij}C_{ti}} &0&0\\
            			0 & 0 &0\\
            			0 & 0 &0
            		\end{bmatrix},
            	\end{equation*} 
            	takes into account the dependence of each local state on the neighboring DGUs.
            	We want to prove \eqref{eq:Lyapeqnoverall}, that, according to the
            	decomposition \eqref{eq:decomposition}, is equivalent to show that
            	\begin{small}
            		\begin{equation}
            			\label{eq:Lyap_abc}
            			\underbrace{\mbf{(\hat{A}_{D}+\hat{B}K)^TP+
            					P(\hat{A}_{D}+\hat{B}K)+}}_{({a})}\underbrace{\mbf{2\hat{A}_{\Xi}P}}_{(b)}+\underbrace{\mbf{\hat{A}_{C}}^T
            				\mbf{P+P\hat{A}_{C}}}_{(c)}\leq 0.
            		\end{equation}
            	\end{small}
            	By means of \eqref{eq:Qi_semidef}, matrix $(a) =
            	\mathrm{diag}(Q_1,\dots,Q_N)$ is negative semidefinite. Now, let us
            	study the contribution of $(b)+(c)$ in \eqref{eq:Lyap_abc}. Matrix $(b)$,
            	by construction, is block
            	diagonal and collects on its diagonal blocks in the form
            	\begin{equation}
            		\label{eq:element_of_b}
            			\begin{aligned}
            				2{\hat{A}}_{\xi i}P_{i}&=
            				\renewcommand\arraystretch{2}
            				\begin{bmatrix}
            					-2\sum\limits_{j\in\NN_i}\frac{1 }{R_{ij}C_{ti}} & 0 & 0& \\
            					0 & 0 & 0 \\
            					0 & 0 & 0\\
            				\end{bmatrix}
            				\left[ \begin{array}{c|c}
            					\eta_i & \mbf{0}_{1\times 2}  \\
            					\hline
            					\mbf{0}_{2\times 1} & \PP_{22,i}
            				\end{array}\right]=                   \\
            				&=    \begin{bmatrix}
            					-2\sum\limits_{j\in\NN_i}\tilde\eta_{ij} & 0 & 0 \\
            					0 & 0 & 0 \\
            					0 & 0 & 0\\
            				\end{bmatrix},
            			\end{aligned}
            	\end{equation}
            	where 
            	\begin{equation}
            		\label{eq:etaij}
            		\tilde\eta_{ij} = \frac{\eta_i }{R_{ij}C_{ti}}.
            	\end{equation}
            	As regards matrix $(c)$, we have that each the block in
            	position $(i,j)$ is equal to  
            	\begin{displaymath}
            		\left\{ \begin{array}{ll}
            			P_i\hat{A}_{ij}+\hat{A}_{ji}^TP_j & \hspace{7mm}\mbox{if } j\in\mathcal{N}_i \\
            			0 & \hspace{7mm}\mbox{otherwise}
            		\end{array} \right.
            	\end{displaymath}
            	where
            	\begin{equation}
            		\label{eq:element_of_c}
            			\begin{aligned}
            				P_i\hat{A}_{ij}+\hat{A}_{ji}^TP_j &=
            				\renewcommand\arraystretch{2}
            				\begin{bmatrix}
            					\tilde\eta_{ij}+\tilde\eta_{ji} & 0 & 0 \\
            					0 & 0 & 0 \\
            					0 & 0 & 0\\
            				\end{bmatrix}.
            			\end{aligned}
            	\end{equation}
            	From \eqref{eq:element_of_b} and \eqref{eq:element_of_c}, we notice
            	that only the elements in position $(1,1)$ of each $3\times 3$
            	block of $(b)+(c)$ can be different from zero. Hence, in order to
            	evaluate the positive/negative definiteness of the $3N\times 3N$
            	matrix $(b)+(c)$, we can equivalently consider the $N\times N$ matrix
            		\begin{equation} 
            			\label{eq:laplacian}
            			\LL = \left[\begin{array}{cccc}
            				-2\sum\limits_{j\in\NN_1}\tilde\eta_{1j} &\bar\eta_{12} & \dots  & \bar\eta_{1N} \\
            				\bar\eta_{21} & \ddots &\ddots & \vdots \\
            				\vdots &\ddots & -2\sum\limits_{j\in\NN_{N-1}}\tilde\eta_{N-1j} & \bar\eta_{N-1N}  \\
            				\bar\eta_{N1}  & \dots  & \bar\eta_{NN-1} & -2\sum\limits_{j\in\NN_N}\tilde\eta_{Nj} 
            			\end{array}
            			\right],
            		\end{equation}
            	obtained by deleting the second and third row and column in each block
            	of $(b)+(c)$. One has $\LL = \MM+\GG$, where
            		\begin{equation*}
            			\mathcal{M}= \begin{bmatrix}
            				-2\sum\limits_{j\in\NN_1}\tilde\eta_{1j} & 0 &\dots &
            				0\\
            				0 & -2\sum\limits_{j\in\NN_2}\tilde\eta_{2j} & \ddots  & \vdots\\
            				\vdots &\ddots& \ddots &0 \\
            				0 & \dots & 0 & -2\sum\limits_{j\in\NN_N}\tilde\eta_{Nj} 
            			\end{bmatrix}
            		\end{equation*}
            	and
            		\begin{equation}
            			\label{eq:GG}
            			\mathcal{G}=\left[\begin{array}{cccc}
            				0&\bar\eta_{12} & \dots  & \bar\eta_{1N} \\
            				\bar\eta_{21} & 0 & \ddots  & \vdots \\
            				\vdots &\ddots & \ddots  & \bar\eta_{N-1N}  \\
            				\bar\eta_{N1}  & \dots & \bar\eta_{N N-1}  & 0
            			\end{array}
            			\right].
            		\end{equation}
            	Notice that each off-diagonal element $\bar\eta_{ij}$ in \eqref{eq:GG} is equal to 
            	\begin{equation}
            		\label{eq:bar_eta}
            		\bar{\eta}_{ij} =\left\{ \begin{array}{ll}
            			(\tilde\eta_{ij}+\tilde\eta_{ji})& \hspace{7mm}\mbox{if } j\in\mathcal{N}_i \\
            			0 & \hspace{7mm}\mbox{otherwise}
            		\end{array} \right.
            	\end{equation}
            	At this point, from Assumption \ref{ass:equal_ratio}, one obtains that
            	$\tilde\eta_{ij}=\tilde\eta_{ji}$ (see \eqref{eq:etaij}) and,
            	consequently, $\bar\eta_{ij} = \bar\eta_{ji}=2\tilde\eta_{ij}$ (see \eqref{eq:bar_eta}). Hence, $\LL$ is
            	symmetric and has non negative off-diagonal elements and zero row and
            	column sum. It follows that $-\LL$ is a Laplacian matrix
            	\cite{grone1990laplacian,godsil2001algebraic}. As such, it verifies $\LL
            	\leq 0$ by construction. Concluding, we have shown that
            	\eqref{eq:Lyap_abc} holds.
            \end{proof}
            The proof of Proposition \ref{pr:semidefinite_abc} reveals that, under
            Assumption \ref{ass:equal_ratio}, interactions between local Lyapunov functions
            $\VV_i(\hat{x}_{[i]})$ due to terms $\hat{A}_{ij}$, $i\neq j$, take
            the form of a weighted Laplacian matrix \cite{grone1990laplacian,godsil2001algebraic} associated to the graph $\GG_{el}$. Furthermore,
            differently from the idea in \cite{Tucci2015a,Tucci2015dc} of nullifying interactions by
            choosing $\eta_i>0$ in \eqref{eq:pstruct} sufficiently small, here \eqref{eq:Lyapeqnoverall} holds
            true even if parameters $\eta_i$ are large.
            
            The next goal is to show asymptotic stability of the mG using the marginal
            stability result in Proposition~\ref{pr:semidefinite_abc} together
            with LaSalle invariance theorem. To this purpose, in Proposition~\ref{pr:quadr_form} we first characterize stationary points for the Lyapunov energy $\VV(\mbf{\hx})$. The main result is then given in Theorem~\ref{thm:overall_stability}.
            This will be done in Theorem \ref{thm:overall_stability},
            that will rely on the
            next Propositions characterizing states $\mathbf{\hat x}$
            yielding $\dot{\mathcal{V}}(\mathbf{\hat x})=0$.
            \begin{prop}
            	\label{prp:ker_F22}
            	Let Assumptions \ref{ass:ctrl} and
            		\ref{ass:equal_ratio} hold and let us define $h_i(v_i) = v_i^T\QQ_{22,i}v_i$, with
            		$v_i\in\mathbb{R}^2$. If \eqref{eq:Qi_semidef} is guaranteed, and $k_{3,i}\neq 0$, then
            		\begin{equation*}
            			h_i(\bar{v}_i)  =0 \Longleftrightarrow\bar{v}_i\in\mathrm{Ker}(\FF_{22,i}).
            	\end{equation*}
            \end{prop}
            \begin{proof}
            	For the sake of simplicity, in the sequel we omit the subscript
            		$i$.
            		We start by proving that
            		\begin{equation}
            			\label{eqn:left}
            			\bar{v}\in\text{Ker}(\FF_{22})\Longrightarrow
            			h(\bar{v})  =0.
            		\end{equation} To this aim, we first replace \eqref{eq:Fi} and \eqref{eq:pstruct}
            		in \eqref{eq:Lyapeqnith}, thus
            		obtaining 
            		\begin{equation}
            			\label{eqn:Q_22}
            			\QQ_{22} =\FF_{22}^T\PP_{22}+\PP_{22}\FF_{22}.
            		\end{equation} 
            		Then, we write
            		\begin{equation*}
            			h(\bar{v}) =\bar{v}^T\QQ_{22}\bar{v} = 
            			2\bar{v}^T\PP_{22}\underbrace{\FF_{22}\bar{v}}_{=0} = 0.
            	\end{equation*}
            	Next, we show that 
            		\begin{equation}
            			\label{eqn:right}
            			h(\bar{v})  =0 \Longrightarrow
            			\bar{v}\in\text{Ker}(\FF_{22}).
            		\end{equation}
            		We start by reformulating the condition $h(\bar{v})  =0$ in \eqref{eqn:right}. In
            		particular, from basic linear algebra, we have the following
            		orthogonal decomposition induced by $\FF_{22}$: $\mathbb{R}^2 = \text{Im}(\FF_{22}^T) \oplus \text{Ker}(\FF_{22})$,
            		which allows us to write any vector $v\in\mathbb{R}^2$ as
            		\begin{equation}
            			\label{eqn:decomposition}
            			v = \hat{v}+\tilde{v}, \hspace{4mm}\hat{v}\in \text{Im}(\FF_{22}^T), \tilde{v}\in \text{Ker}(\FF_{22}).
            		\end{equation}
            		Since we are assuming that $Q$ is negative semidefinite and
            		structured as in \eqref{eq:qstruct}, vectors $\bar v$ satisfying
            		$h(\bar{v})  =0$ also maximize $h(\cdot)$. Hence, 
            		\begin{equation}
            			\label{eqn:maximum}
            			h(\bar{v})  =0 \Longleftrightarrow\frac{\text{d}h}{\text{d}v}(\bar v) = \QQ_{22}\bar v = 0,
            		\end{equation}
            		which,
            		decomposing $\bar v$ as in \eqref{eqn:decomposition}, provides
            		\begin{equation}
            			\label{eqn:star_formulation}
            			h(\bar{v})  =0 \Longleftrightarrow \QQ_{22}\hat{v}+ \underbrace{\QQ_{22}\tilde{v}}_{=0} =0.
            		\end{equation}
            		Notice that $\QQ_{22}\tilde{v}=0$ in \eqref{eqn:star_formulation} follows
            		from the fact that $\tilde{v}\in \text{Ker}(\FF_{22})$. In particular,
            		from \eqref{eqn:left}, we know that $h(\tilde v) = 0$, and hence condition
            		\eqref{eqn:maximum} must hold for $\bar v = \tilde v$.
            		At this point, using
            		\eqref{eqn:star_formulation}, we can
            		rewrite \eqref{eqn:right} as  
            		\begin{equation*}
            			\QQ_{22}\hat{v}=0 \Longrightarrow \bar v \in \text{Ker}(\FF_{22}),
            		\end{equation*}
            		which, since $\bar{v}\in\text{Ker}(\FF_{22}) \Longleftrightarrow \hat v
            		= 0$, finally becomes
            		\begin{equation}
            			\label{eqn:final_reformulation}
            			\QQ_{22}\hat{v}=0 \Longrightarrow \hat v =0.
            		\end{equation}
            		In summary, we have shown that, in order to prove \eqref{eqn:right}, one can equivalently demonstrate \eqref{eqn:final_reformulation}. To this aim, we parametrize $\hat{v}\in\text{Im}(\FF_{22}^T)$ as
            		\begin{equation*}
            			\text{Im}(\FF_{22}^T)=\{\FF_{22}^T
            			\begin{bmatrix}
            				y_1\\
            				y_2 \\
            			\end{bmatrix}
            			, y_1, y_2\in\mathbb{R}\}
            		\end{equation*}
            		which, recalling \eqref{eq:Fi}, becomes
            		\begin{equation*}
            			\text{Im}(\FF_{22}^T)=\left\lbrace\begin{bmatrix}
            				f_{22} & 0\\
            				f_{23} & 0\\
            			\end{bmatrix}
            			\begin{bmatrix}
            				y_1\\
            				y_2 \\
            			\end{bmatrix}
            			, y_1, y_2\in\mathbb{R}\right\rbrace=\left\lbrace
            			\begin{bmatrix}
            				f_{22}\\
            				f_{23} \\
            			\end{bmatrix}y_1, y_1\in\mathbb{R}\right\rbrace.
            		\end{equation*}
            		Hence, we rewrite $\QQ_{22}\hat{v}=0$ in \eqref{eqn:final_reformulation} as $\QQ_{22}[f_{22}\text{ }f_{23}]^{T}y_1=0$, that, by means of \eqref{eqn:Q_22}, implies
            		\begin{equation}
            			\label{eqn:sys_parametrization}
            			\begin{small}\PP_{22}\FF_{22}
            				\begin{bmatrix}
            					f_{22}  \vspace{2mm}\\
            					f_{23} \\
            				\end{bmatrix}y_1
            				=-\FF_{22}^T\PP_{22}
            				\begin{bmatrix}
            					f_{22} \vspace{2mm}\\
            					f_{23} \\
            				\end{bmatrix}y_1.
            			\end{small}
            		\end{equation}
            		Replacing \eqref{eq:Fi} and \eqref{eqn:P22} in \eqref{eqn:sys_parametrization}, we get
            		\begin{equation}
            			\label{eq:sys_prop3}            
            			\left\lbrace
            			\begin{aligned}
            				p_{22}\left(f_{22}^2+f_{23}^2\right)y_1&=-f_{22}^2p_{22}y_1-f_{22}f_{23}p_{23}y_1 \\
            				p_{23}\left(f_{22}^2+f_{23}^2\right)y_1&=-f_{22}f_{23}p_{22}y_1+f_{23}^2p_{23}y_1  
            			\end{aligned}
            			\right.
            		\end{equation}
            		Notice that \eqref{eq:sys_prop3} is verified if $y_1 =0$ (i.e. if $\hat v = 0$). To conclude the proof, we just need to show that $y_1 =0$ is the only solution of \eqref{eq:sys_prop3}. To this purpose, we proceed by contradiction and assume that there exists a $y_1 \neq 0$ fulfilling \eqref{eq:sys_prop3}. This leads to
            		\begin{subequations}
            			\label{eq:sys2_prop3}            
            			\begin{empheq}[left=$\quad$\empheqlbrace]{align}
            				\label{eq:sysA_prop3}
            				p_{22}\left(2f_{22}^2+f_{23}^2\right)&=-f_{22}f_{23}p_{23} \\
            				\label{eq:sysB_prop3}
            				p_{23}\left(f_{22}^2+f_{23}^2\right)&=-f_{22}f_{23}p_{22}+f_{23}^2p_{23}  
            			\end{empheq}
            		\end{subequations}
            		Let us consider \eqref{eq:sysA_prop3} and let us assume, for the time being, 
            		\begin{equation}
            			\label{eqn:ass_f22_f23}
            			f_{22}\neq 0 \text{ and } f_{23}\neq 0.
            		\end{equation}
            		We will show later that these two conditions are satisfied if $k_3\neq 0$. Since it also holds $p_{22}>0$\footnote{Matrix $P$ is structured as in \eqref{eq:pstruct} and it is positive definite. Therefore, $\PP_{22}>0$, which implies that the first minor of $\PP_{22}$ (i.e. $p_{22}$) must be strictly positive.}, we have $-f_{22}f_{23}p_{23}>0$, which implies $p_{23}\neq 0$. Then, we derive
            		\begin{equation*}
            			p_{22}=\frac{f_{22}f_{23}p_{23}}{2f_{22}^2+f_{23}^2}
            		\end{equation*}
            		and replace it in \eqref{eq:sysB_prop3}, obtaining
            		\begin{equation}
            			\label{eqn:prop3_finale}
            			2p_{23}f_{23}^2+p_{23}f_{22}^2=\frac{f_{22}^2f_{23}^2p_{23}}{2f_{22}^2+f_{23}^2}.
            		\end{equation}
            		Finally, by direct calculation, \eqref{eqn:prop3_finale} amounts to
            		\begin{equation*}
            			4f_{22}^2f_{23}^2+2f_{23}^4+2f_{22}^4=0,
            		\end{equation*}
            		which is true if and only if $f_{22}=f_{23}=0$. However, from \eqref{eqn:ass_f22_f23}, these conditions are never verified.
            	
		            The last step is to show that \eqref{eqn:ass_f22_f23} holds. Recalling that electrical parameters are positive, one has $k_{3} \neq 0 \Longrightarrow f_{23} \neq 0$ (see \eqref{eq:Fi}). Moreover, $Q\leq 0$, implies $\QQ_{22}\leq 0$ (in fact, as the last row of $\FF_{22}$ is zero, $\QQ_{22}$ cannot be negative definite). Now, if $f_{22} = 0$ holds, we would have $q_{22}=0$, which implies $q_{23}=q_{32}=0$. However, by construction, we have $q_{23}=q_{32}=f_{23}p_{22} + f_{22}p_{23}$ (see \eqref{eq:qstruct}), which is never  zero if $f_{22,i} = 0$, because $f_{23} \neq 0$ and $p_{22} \neq 0$.
            \end{proof}
  
            \begin{prop}
            	\label{pr:quadr_form}
            	Let $g_i(w_i) = w_i^T Q_i w_i$. Under the same assumptions of Proposition \ref{prp:ker_F22}, $\forall i\in\DD$, only vectors $\bar w_i$ in the form
            	\begin{equation*}
            		\bar{w}_i =\left[ \begin{array}{ccc}
            			\alpha_i&
            			\beta_i  &
            			\delta_i\beta_i
            		\end{array}\right]^T
            	\end{equation*} with $\alpha_i$, $\beta_i\in\Rset$, and $\delta_i =
            	-\frac{k_{2,i}-R_{ti}}{k_{3,i}}$, fulfill
            	\begin{equation}
            		\label{eq:vQv=0}
            		g_i(\bar{w}_i) = \bar{w}_i^T Q_i \bar{w}_i=0.
            	\end{equation}
            \end{prop}
            \begin{proof}
            	In the sequel, we omit the subscript
            	$i$. From \eqref{eq:qstruct}, $g(w)$ is equal to 
            	\begin{equation}
            		 \label{eq:quadr_form}
            		\left[ \begin{array}{c|c} 
            			w_1 & {w}_2^T
            		\end{array}
            		\right]
            		\left[
            		 \renewcommand\arraystretch{1.5} \begin{array}{c|c}
            			0 & \mbf{0}_{1\times 2}  \\
            			\hline
            			\mbf{0}_{2\times 1} & \QQ_{22}\\
            		\end{array}\right]\left[ \begin{array}{c}
            			w_1  \\
            			\hline
            			{w}_2
            		\end{array}\right],
            	\end{equation} 
            	where $ w_2\in\Rset^2$. Since $Q$ is
            	negative semidefinite, the 
            	vectors $\bar w$ satisfying \eqref{eq:vQv=0} also maximize $g(\cdot)$. Hence, it must hold $ \frac{\mathrm{d}g}{\mathrm{d}w}(\bar w)
            		= Q\bar w=0$, i.e.
            	\begin{equation}
            		\label{eq:maximum}
            		\left[ 
            		 \renewcommand\arraystretch{1.5}
            		 \begin{array}{c|cc}
            			0 &  \mbf{0}_{1\times 2}   \\
            			\hline
            			\mbf{0}_{2\times 1}  & \QQ_{22}\\
            		\end{array}\right]\left[ \begin{array}{c}
            			\bar{w}_1  \\
            			\hline
            			\bar{{w}}_2
            		\end{array}\right]=0.
            	\end{equation}
            	It is easy to show that, by direct calculation, a set of solutions
            	to \eqref{eq:vQv=0} and \eqref{eq:maximum} is composed of vectors in the form
            	\begin{equation}
            		\label{eq:first_solution}
            		\bar{w} = \left[ \begin{array}{ccc}
            			\alpha &
            			0&
            			0
            		\end{array}\right]^T ,\quad\alpha\in\Rset.
            	\end{equation}
            	Moreover, from \eqref{eq:quadr_form}, we have that \eqref{eq:vQv=0} is
            	also verified if there exist vectors 
            	\begin{equation}
            		\label{eq:tilde_v}
	            	\tilde w =\left[ \begin{array}{c|c}
            				w_1 &
            				{\underline{w}}_2^T
            			\end{array}\right]^T,\quad{\underline{w}}_2\neq [0\text{ }0]^T,
            	\end{equation}
            	such that $w_1\in\Rset$ and
            	\begin{equation}
            		\label{eq:second_solution}
            		{\underline{w}}_2^T
            			\QQ_{22} {\underline{w}}_2 =0.
            	\end{equation}
            	By exploiting the result of Proposition \ref{prp:ker_F22}, we know that vectors $\underline{w}_2$ fulfilling \eqref{eq:second_solution} belong to $\text{Ker}(F_{22})$, which, recalling \eqref{eq:Fi}, can be explicitly computed as follows
            		\begin{equation}
            			\label{eqn:Ker_F_22}
            			\begin{aligned}
            				\mathrm{Ker}(\FF_{22}) &=\left\lbrace x\in\Rset^2 : \left[ \begin{array}{cc}
            					f_{22}  & 0\\
            					f_{23}  & 0
            				\end{array}\right]x = 0\right\rbrace =\\
            				&=\left\lbrace x\in\Rset^2 : x = [\text{ } \beta\text{ }\delta\beta\text{ }]^T, \beta\in\Rset,\delta = -\frac{k_{2}-R_{t}}{k_{3}} \right\rbrace .
            			\end{aligned}
            	\end{equation}
            	The proof ends by merging \eqref{eq:first_solution} and \eqref{eq:tilde_v}, with ${\underline{w}}_2$ as in \eqref{eqn:Ker_F_22}.
            
            \end{proof}
 
            \begin{thm}
            	\label{thm:overall_stability}
            	If Assumptions \ref{ass:ctrl} and \ref{ass:equal_ratio} are fulfilled,
            	the graph $\GG_{el}$ is connected, \eqref{eq:Qi_semidef} holds and
            	$k_{3,i}\neq 0$, $\forall i\in\DD$, the origin of \eqref{eq:sysaugoverallclosed} is asymptotically stable.
            \end{thm}
            \begin{proof}
            	From Proposition \ref{pr:semidefinite_abc}, $\dot\VV(\mbf{\hat x})$ is negative semidefinite
            	(i.e. \eqref{eq:Lyapeqnoverall} holds). 
            	We want to show that the origin of the mG is also attractive using the LaSalle
            	invariance Theorem \cite{khalil2001nonlinear}. For this purpose, we first compute the set
            	$R = \{\mbf{x}\in\Rset^{3N} : \mbf{x}^T \mbf{Q}\mbf{x}= 0 \}$, which,
            	by means of the decomposition in \eqref{eq:Lyap_abc}, coincides with
            	\begin{equation}
            		\label{eq:R}
            		\begin{aligned}
            			R &= \{\mbf{x} : \mbf{x}^T \left( (a)+(b)+(c)\right)\mbf{x}
            			= 0 \}\\
            			&=\{\mbf{x} : \mbf{x}^T (a) \mbf{x} +\mbf{x}^T(b)\mbf{x}+\mbf{x}^T(c) \mbf{x}
            			= 0 \}\\
            			&=\underbrace{\{\mbf{x} : \mbf{x}^T (a) \mbf{x} =0\}}_{X_1}\cap \underbrace{\{\mbf{x}:\mbf{x}^T\left[(b)+(c)\right] \mbf{x} =0\}}_{X_2} .
            		\end{aligned}
            	\end{equation}
            	In particular, the last equality follows from the fact that $(a)$ and $(b)+(c)$ are negative semidefinite matrices (see the proof of Proposition \ref{pr:semidefinite_abc}).
            	We first focus on the elements of set $X_2$. Since matrix $(b)+(c)$ can
            	be seen as an "expansion" of the Laplacian \eqref{eq:laplacian}, with zero entries
            	on the second and third row of each $3\times 3$ block, we have that,
            	by construction, vectors in the form
            	\begin{small}
            		\begin{equation}
            			\label{eq:X2_a}
            			\mbf{\tilde x} =\left[  \text{ }0\text{ }\tilde x_{12} \text{ }\tilde
            			x_{13} \text{ }|\text{ }\cdots \text{ }|\text{ }0\text{ }\tilde
            			x_{N2} \text{ }\tilde x_{N3} \text{ }\right]^T,\hspace{3mm}\tilde
            			x_{i2}, \tilde x_{i3}\in\Rset,\forall i\in\DD,
            		\end{equation}
            	\end{small}
            	\normalsize
            	belong to $X_2$.
            	Moreover, since the kernel of the Laplacian matrix of a connected graph
            	contains only vectors with identical entries \cite{godsil2001algebraic}, it also holds
            	\begin{equation}
            		\label{eq:X2_b}
            		\{\mbf{\bar{x}}=\bar{x}\left[  \text{ }1\text{ }0 \text{ } 0 \text{
            		}|\text{ }\cdots \text{ }|\text{ }1\text{ } 0 \text{ } 0 \text{
            		}\right]^T,~\bar x \in \Rset\}\subset X_2,
            	\end{equation}
            	with $\bar x\in\Rset$. Hence, by merging \eqref{eq:X2_a} and \eqref{eq:X2_b}, we have that
            	\begin{small}
            		\begin{equation*}
            			X_2 = \{\mbf{x}:\mbf{x} = \left[  \text{ }\bar{x}\text{ }\tilde x_{12}
            			\text{ } \tilde x_{13} \text{ }|\text{ }\cdots \text{ }|\text{ }\bar{x}\text{ }
            			\tilde x_{N2} \text{ } \tilde x_{N3} \text{ }\right]^T, \bar{x},
            			\tilde{x}_{i2},\tilde{x}_{i3}\in\Rset\}.
            		\end{equation*}
            	\end{small}
            	Next, we characterize the set $X_1$. By exploiting Proposition
            	\ref{pr:quadr_form}, it follows that
            	\begin{small}
            		\begin{equation}
            			X_1 = \{\mbf{x}:\mbf{x} =\left[  \text{ }\alpha_1\text{ }\beta_1
            			\text{ } \delta_1\beta_1 \text{ }|\text{ }\cdots \text{ }| \text{ }\alpha_N\text{ }\beta_N
            			\text{ } \delta_N\beta_N \text{ }\right]^T, \alpha_i,\beta_i\in\Rset\},
            		\end{equation}
            	\end{small}
            	and, from \eqref{eq:R},
            	\begin{small}
            		\begin{equation}
            			\label{eq:R_new}
            			R = \{\mbf{x}:\mbf{x} =\left[  \text{ }\bar{\alpha}\text{ }\beta_1
            			\text{ } \delta_1\beta_1 \text{ }|\text{ }\cdots \text{ }| \text{ }\bar{\alpha}\text{ }\beta_N
            			\text{ } \delta_N\beta_N \text{ }\right]^T, \bar{\alpha}, \beta_i\in\Rset\}.
            		\end{equation}
            	\end{small}
            	At this point, in order to conclude the proof, we need to show that the largest
            	invariant set $M\subseteq R$ is the origin. To this purpose, we consider \eqref{eq:modelDGUgen-aug-closed}, include coupling terms $\subss{\hat\xi}i$, set $\hat d_{[i]}= 0$ and choose as initial
            	state $\mbf{\hat x}(0) = \left[ \hat x_1(0)|\dots|\hat
            	x_N(0)\right]^T\in R$. We aim to find conditions on the
            	elements of  $\mbf{\hat x}(0)$ that must hold for having
            	$\mbf{\dot{\hat{x}}}\in R$. One has
            	\begin{equation*}
            		\begin{aligned}
            			\dot{\hat x}_i(0) &={F_i}\hat
            			x_i(0)+\sum\limits_{j\in\NN_i}\underbrace{\hat A_{ij}\left(\hat
            				x_j(0)-\hat x_i(0)\right)}_{=0}=\\
            			&=\left[\begin{array}{ccc}
            				0 & \frac{1}{C_{ti}} & 0 \vspace{2mm}\\ 
            				\frac{k_{1,i}-1}{L_{ti}} & \frac{k_{2,i}-R_{ti}}{L_{ti}} & \frac{k_{3,i}}{L_{ti}} \vspace{2mm}\\
            				-1 & 0 & 0\\
            			\end{array}\right]\left[\begin{array}{c}
            				\bar\alpha \vspace{2mm}\\
            				\beta_i \vspace{2mm}\\
            				\delta_i\beta_i
            			\end{array}\right]=
            			\\
            			&=\left[\begin{array}{c}
            				\frac{\beta_i}{C_{ti}}  \vspace{2mm}\\
            				\frac{k_{1,i}-1}{L_{ti}}\bar{\alpha} +\underbrace{\frac{k_{2,i}-R_{ti}}{L_{ti}}\beta_i+\frac{k_{3,i}}{L_{ti}}\delta_i\beta_i}_{=0}\vspace{2mm}\\
            				-\bar{\alpha}
            			\end{array}\right]=\left[\begin{array}{c}
            				\frac{\beta_i}{C_{ti}}  \vspace{2mm}\\
            				\frac{k_{1,i}-1}{L_{ti}}\bar{\alpha} \vspace{2mm}\\
            				-\bar{\alpha}
            			\end{array}\right],
            			\normalsize
            		\end{aligned}
            	\end{equation*}
            	for all $i\in\DD$. From \eqref{eq:R_new}, $\mbf{\dot{\hat{x}}}(0)\in R$ if and only if, it holds
            	\begin{subequations}
            		\begin{empheq}[left =  \empheqlbrace]{align}
            			\label{eq:R_inv_a} \frac{\beta_i}{C_{ti}} &=
            			\bar\rho\\
            			\label{eq:R_inv_b}\frac{k_{1,i}-1}{L_{ti}}\delta_i\bar\alpha&=-\bar\alpha,
            		\end{empheq}
            	\end{subequations}
            	for $i\in\DD$ and $\bar\rho\in\Rset$. Condition \eqref{eq:R_inv_b} is fulfilled if either
            	\begin{equation}
            		\label{eq:alpha_0}
            		\bar\alpha = 0
            	\end{equation}
            	or 
            	\begin{equation}
            		\label{eq:delta_solutions}
            		\delta_i = -\frac{L_{ti}}{k_{1,i}-1}.
            	\end{equation}
            	Let us
            	first focus on \eqref{eq:delta_solutions}. From Proposition \ref{pr:prop_1}, we have that
            	$q_{12,i}=q_{13,i} = 0$. By direct computation, one has that
            	\begin{equation}
            		\label{eq:p12_0}
            		q_{12,i} = \frac{(k_{1,i}-1)}{L_{ti}}p_{22,i} -p_{23,i}
            		+\frac{\eta_i}{C_{ti}}.              
            	\end{equation}
            	Moreover, we show that 
            		\begin{equation}
            			\label{eq:conditionp22}
            			p_{22,i}=-\delta p_{23,i}.
            	\end{equation}
            	\label{rmk:elements_of_P}
            		From Proposition \ref{prp:ker_F22}, we know that only vectors $v\in\mathrm{Ker}(\FF_{22})$ satisfy \eqref{eqn:maximum}. Then, \eqref{eq:conditionp22} can be obtained by solving \eqref{eqn:maximum} with vectors defined as in \eqref{eqn:Ker_F_22}.
            	
            	By exploiting \eqref{eq:conditionp22}, substituting \eqref{eq:equal_ratio} in \eqref{eq:p12_0}, and setting
            	$q_{12,i}=0$, we get
            	\begin{equation}
            		\label{eq:k1}
            		k_{1,i} = 1-\frac{L_{ti}}{\delta_i}-\bar\sigma\frac{L_{ti}}{p_{22,i}}.
            	\end{equation}
            	Then, if we replace \eqref{eq:k1} in \eqref{eq:delta_solutions}, we obtain
            	\begin{equation*}
            			\delta_i =
            			\frac{\delta_i p_{22,i}}{p_{22,i}+\bar\sigma\delta_i}=\frac{\delta_i}{1+\bar\sigma\frac{\delta_i}{p_{22,i}}},
            	\end{equation*}
            	which is true if $\bar\sigma\frac{\delta_i}{p_{22,i}}=0$. This latter condition, however, is never verified since $L_{ti}$ in \eqref{eq:delta_solutions} is always positive, as well as $p_{22,i}$ and $\bar\sigma$. It follows that
            	\eqref{eq:R_inv_b} has only one solution, which is
            	\eqref{eq:alpha_0}. Therefore, by considering also the solutions of \eqref{eq:R_inv_a},
            	we find that $\mbf{\dot{\hat{x}}}(0)\in R$ only if
            	$\mbf{\hat{x}}(0)\in S$, where 
            	\begin{small}
            		\begin{equation}
            			\label{eq:S}
            			S= \{\mbf{x} = \left[\text{ }0\text{ }\bar\rho C_{t1}\text{ }
            			\bar\rho\delta_1C_{t1}\text{ }|\text{ }\dots\text{ }|\text{ }0\text{
            			}\bar\rho C_{tN}\text{ }\bar\rho\delta_NC_{tN}\text{
            			}\right]^T, \bar{\rho}\in\Rset\}.
            		\end{equation}
            	\end{small}
            	Furthermore, it must hold $M \subseteq S$. Then, in order to
            	characterize $M$, we pick an initial state $\mbf{\tilde x}(0) = \left[ \tilde x_1(0)|\dots|\tilde
            	x_N(0)\right]^T\in S$ and impose
            	$\mbf{\dot{\tilde{x}}}(0)\in S$. This translates into the equations 
            		\begin{equation*}
            			\begin{aligned}
            				\dot{\tilde x}_i(0) &=(\hat{A}_{ii}+\hat B_iK_i)\tilde
            				x_i(0)+\sum\limits_{j\in\NN_i}\underbrace{\hat A_{ij}\left(\tilde x_j(0)-\tilde x_i(0)\right)}_{=0}=\\
            				&=(\hat{A}_{ii}+\hat B_iK_i)\left[\begin{array}{c}
            					0\\
            					\bar\rho C_{ti}\\
            					\bar\rho\delta_iC_{ti}
            				\end{array}\right] = 
            				\left[\begin{array}{c}
            					{\bar\rho}  \vspace{2mm}\\
            					0 \vspace{2mm}\\
            					0
            				\end{array}\right],
            			\end{aligned}
            		\end{equation*}
            	for all $i\in\DD$. It follows that $\mbf{\dot{\tilde{x}}}(0)\in S$ only if
            	$\bar\rho = 0$. Since $M\subseteq S$, from \eqref{eq:S} one has $M = \{0\}$.
            \end{proof}
            
            \subsection{Line-independent controller computation through LMIs}
            \label{subsec:LMI}
            We now show how to compute matrices $K_i$ and $P_i$ via numerical optimization so as to comply with assumptions of Theorem~\ref{thm:overall_stability}.
            In order to enforce, when possible, a margin of robustness, controllers $K_i$ should be designed such that inequality   
            \begin{equation}
            	\label{eq:Lyapdecr}
            	(\hat{A}_{ii}+\hat{B}_{i}K_{i})^{T}P_{i}+P_{i}(\hat{A}_{ii}+\hat{B}_{i}K_{i})+\Gamma_{i}^{-1}\le 0,
            \end{equation} with $\Gamma_i =
            \mathrm{diag}(\gamma_{1i},\gamma_{2i},\gamma_{3i})$, is
            verified for $\gamma_{ki}\geq 0$, $k = 1,2,3,$ and matrix $P_{i}$ structured as in \eqref{eq:pstruct}. 
            The design of the local
            controller $\subss{\CC}{i}$ is performed solving the following
            problem.
            \begin{prob}
            	\label{prbl:designPrbl}
            	For parameters $\eta_i$ chosen as in \eqref{eq:equal_ratio}, compute a vector $K_{i}$ such that 
            	Assumption \ref{ass:ctrl} is verified and \eqref{eq:Qi_semidef} holds.
            \end{prob}
            Consider the following optimization problem
            \begin{subequations}
            	\label{eq:optproblem}
            		\begin{align}
            			\mathcal{O}_i: &\min_{\substack{Y_{i},G_{i},\gamma_{1i},\\\gamma_{2i},\gamma_{3i},\beta_{i},\zeta_{i}}}\quad \alpha_{1i}\gamma_{1i}+\alpha_{2i}\gamma_{2i}+\alpha_{3i}\gamma_{3i}+\alpha_{4i}\beta_{i}+\alpha_{5i}\zeta_{i}\nonumber \\
            			\label{eq:Ystruct}&Y_{i}=\left[ \begin{smallmatrix}
            				\eta_i^{-1} & 0 & 0 \\
            				0 & \bullet & \bullet\\
            				0 & \bullet & \bullet\\
            			\end{smallmatrix}\right]>0\\
            			\label{eq:LMIstab}&\begin{bmatrix}
            				Y_{i}\hat{A}_{ii}^{T}+G_{i}^{T}\hat{B}_{i}^{T}+\hat{A}_{ii}Y_{i}+\hat{B}_{i}G_{i} & Y_{i} \\
            				Y_{i} & -\Gamma_{i}
            			\end{bmatrix}\le 0\\
            			\label{eq:Gcostr}&\begin{bmatrix}
            				-\beta_{i}I & G_{i}^T\\
            				G_{i} & -I
            			\end{bmatrix}<0\\
            			\label{eq:Ycostr}&\begin{bmatrix}
            				Y_{i} & I\\
            				I & \zeta_{i}I
            			\end{bmatrix}>0\\
            			&\gamma_{1i},\gamma_{2i},\gamma_{3i}\geq 0,\quad\beta_{i}>0,\quad\zeta_{i}>0
            		\end{align}
            \end{subequations}  
            where $\alpha_{ji}$, $j=1\dots,5$ represent positive weights and $\bullet$ are arbitrary
            entries. We notice that all constraints in \eqref{eq:optproblem}
            are LMIs. Therefore, the optimization
            problem $\OO_i$ is convex and can be solved in
            polynomial time \cite{Boyd1994}. 
            The next Lemma, proved in \cite{Tucci2015a}, establishes the relations between problem $\OO_i$ and matrices $K_i$ and $P_i$. 
            \begin{lem}
            	\label{lem:optProbl}
            	Problem $\mathcal{O}_i$ is feasible if and only if
            	Problem \ref{prbl:designPrbl} has a
            	solution. Moreover, $K_i$ and $P_i$ in
            	\eqref{eq:Lyapeqnith} are given by
            	$K_{i}=G_{i}Y_{i}^{-1}$, $P_{i}=Y_{i}^{-1}$ and 
            	\begin{equation}
            		\label{eq:K_magn}
            		\norme{K_{i}}{2}<\sqrt{\beta_i}\zeta_{i}.
            	\end{equation}
            \end{lem}
            We highlight that a suitable tuning of weights $\alpha_{ji}$,
            $j=1,\dots,5$, in the cost of problem $\OO_i$ allows one to achieve a
            trade-off between the magnitude of $\Gamma_i$ and the
            aggressiveness of the control action, governed in \eqref{eq:K_magn} by the
            magnitude of $\beta_i$ and $\zeta_i$.
            
            Next, we discuss the key features of the proposed
            decentralized control approach. We first notice that
            constraints in
            \eqref{eq:optproblem} depend upon local fixed matrices
            ($\hat{A}_{ii}$, $\hat{B}_{i}$) and local design
            parameters ($\alpha_{1i}$, $\alpha_{2i}$,
            $\alpha_{3i}$, $\alpha_{4i}$, $\alpha_{5i}$). It follows that the computation of controller
            $\subss{\CC}{i}$ is completely independent from the
            computation of controllers $\subss{\CC}{j}$, $j\neq
            i$, up to the knowledge of the common parameter $\bar\sigma$ in \eqref{eq:equal_ratio}. Secondly, \eqref{eq:optproblem} is independent of parameters of
            electrical lines connecting DGUs. Thirdly, as discussed
            after Proposition \ref{pr:semidefinite_abc}, differently from
            \cite{Tucci2015a,Tucci2015dc}, the design procedure does not require
            that parameters $\eta_i$ are sufficiently small, so as
            to reduce the coupling among DGUs (see term (b) in (23)
            of \cite{Tucci2015a}). Finally, if problems
            $\OO_i$, $i\in\DD$, are
            feasible, then the
            overall closed-loop mG is asymptotically stable,
            provided that $k_{3,i}\neq 0$, $i\in\DD$ (see Theorem
            \ref{thm:overall_stability}). As for the latter condition, it was always
            fulfilled in all numerical experiments we
            performed. Nevertheless, if it does not hold, one can
            try to solve again \eqref{eq:optproblem} after changing the weights
            in the cost.
            Algorithm
            \ref{alg:ctrl_design} collects the steps of the overall
            design procedure. For improving the closed-loop
            bandwidth of each controlled DGU and the rejection of
            current disturbances $I_{Li}$, it also includes the
            optional design of pre-filters $\tilde{C}_{[i]}$ and disturbance
            compensators $N_{[i]}$ \cite{Tucci2015a,Tucci2015dc}.
            \begin{algorithm}[!htb]
            	\caption{Design of controller $\subss{\CC}{i}$ and compensators $\subss{\tilde{C}}{i}$ and $\subss{N}{i}$ for subsystem $\subss{\hat{\Sigma}}{i}^{DGU}$}
            	\label{alg:ctrl_design}
            	\textbf{Input:} DGU $\subss{\hat{\Sigma}}{i}^{DGU}$ as in \eqref{eq:modelDGUgen-aug} \\
            	\textbf{Output:} Controller $\subss{\CC}{i}$ and, optionally, pre-filter $\subss{\tilde{C}}{i}$ and compensator $\subss{N}{i}$\\
            	\begin{enumerate}[(A)]
            		\item\label{enu:stepAalgCtrl} Find $K_i$ solving the
            		LMI problem \eqref{eq:optproblem}. If it is not
            		feasible, or $k_{3,i}\neq 0$ cannot be obtained, \textbf{stop} (the controller $\subss\CC i$ cannot be designed).\\
            		\textbf{Optional step}
            		\item\label{enu:stepBalgCtrl} Design an
            		asymptotically stable local pre-filter
            		$\subss{\tilde{C}}{i}$ and compensator
            		$\subss{N}{i}$ (see Section III.D in \cite{Tucci2015a}).
            	\end{enumerate}
            \end{algorithm}
            
            \begin{rmk}
            		\label{rmk:}
            		In order to assess the conservativeness of the LMI \eqref{eq:optproblem}, we solved it for $\bar\sigma = 10$ and for various combinations of parameters $(R_t, L_t, C_t)$ characterizing the DGUs, checking when they are infeasible. In particular, we derived from the literature meaningful parameter ranges for converters typically used in low-voltage DC mGs. Numerical results, which are reported in Appendix B, show that the LMIs are always feasible.
            \end{rmk}

\subsection{PnP operations}
    \label{sec:PnP}
As in \cite{Tucci2015a,Tucci2015dc}, we describe the 
operations that need to be performed for handling the
addition and removal of DGUs, while
preserving the stability of the overall closed-loop
system. We consider an mG composed of subsystems
$\subss{\hat{\Sigma}}{i}^{DGU}, i\in\DD$, equipped with
local controllers $\subss{\CC}{i}$ and compensators
$\subss{\tilde{C}}{i}$ and $\subss{N}{i}$, $i\in\DD$
produced by Algorithm \ref{alg:ctrl_design}. We also
assume the graph $\GG_{el}$ is connected.

\textbf{Plugging-in operation}
Assume that a new DGU
$\subss{\hat{\Sigma}}{N+1}^{DGU}$
sends
a plug-in request. Let $\NN_{N+1}$ be the
set of DGUs that will be connected to
$\subss{\hat{\Sigma}}{N+1}^{DGU}$ through
power lines. The
design of controller $\subss{\CC}{N+1}$ and compensators
$\subss{\tilde{C}}{N+1}$ and $\subss{N}{N+1}$
requires Algorithm \ref{alg:ctrl_design} to be
executed. Differently from the plugging-in
protocol described in \cite{Tucci2015a,Tucci2015dc}, there is no need to 
redesign controllers $\subss{\CC}{j}$ and
compensators $\subss{\tilde{C}}{j}$ and
$\subss{N}{j}$, $\forall j\in\NN_{N+1}$, because
matrices $\hat{A}_{jj}$, $j\in\NN_{N+1}$ do not
change. Therefore, if Algorithm \ref{alg:ctrl_design} does not stop in Step
\ref{enu:stepAalgCtrl} when computing controllers $\subss{\CC}{N+1}$,
the plug-in of $\subss{\hat{\Sigma}}{N+1}^{DGU}$ is allowed.

\textbf{Unplugging operation}
Assume now that DGU $k\in\DD$, needs
to be disconnected from the network. Differently
from \cite{Tucci2015a,Tucci2015dc}, since the unplugging of
subsystem $\subss{\hat{\Sigma}}{k}^{DGU}$
does not change the matrix
$\hat{A}_{jj}$ of each
$\subss{\hat{\Sigma}}{j}^{DGU}$, $
j\in\NN_{k}$, DGU $k$
can be removed without redesigning the local
controllers $\subss\CC j$,
$j\notin\NN_{k}$. in view of
Theorem~\ref{thm:overall_stability}, stability is
preserved as long as the new graph $\GG_{el}$ is still connected.
\begin{rmk}
	According to the above PnP operations, whenever a DGU $i$ wants to be
	plugged -in or -out, no updating of controllers of neighboring DGU $j$,
	$j\in\NN_i$ is required. As a consequence, there is no
	need to equip each local
	controller with bumpless control scheme described in \cite{Tucci2015a,Tucci2015dc} for
	ensuring smooth behaviors of the control variable when controllers are
	switched in real-time. 
\end{rmk}
 \section{Simulation results}
                \label{sec:scenario_2}
  In this section, performance of the proposed controllers is evaluated. In order to compare
  the new PnP design methodology with the one in
  \cite{Tucci2015a,Tucci2015dc}, we performed the same simulation described
  in Scenario 2 in \cite{Tucci2015a}. Notably, we consider the meshed mG
  in Figure~\ref{fig:5areasplug} 
  composed of  DGUs $1-5$. 
  We highlight that the DGUs have non-identical electrical
  parameters, which are listed in Tables 2, 3
  and 4 in Appendix C of \cite{Tucci2015dc}. As in \cite{Tucci2015a}, voltage
  references for the DGUs are set to slightly different values (see Table II, in \cite{Tucci2015a}), so as to
  make the case study more realistic, whereas
  the constant ratio $\bar\sigma$ in \eqref{eq:equal_ratio} has been chosen equal to 10.
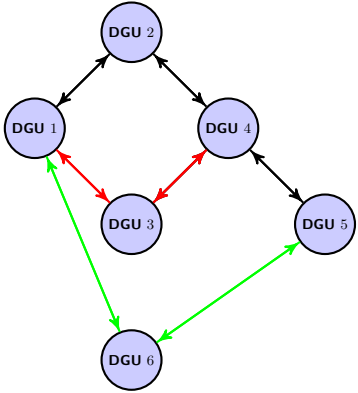
\begin{figure}[!htb]
                 \centering
                 	\begin{tikzpicture}[scale=0.6,transform shape,->,>=stealth',shorten >=1pt,auto,node distance=3cm, thick,main node/.style={circle,fill=blue!20,draw,font=\sffamily\bfseries}]
					  
  \node[main node] (1) {\small{DGU $1$}};
  \node[main node] (2) [above right of=1] {\small{DGU $2$}};
  \node[main node] (3) [below right of=1] {\small{DGU $3$}};
  \node[main node] (4) [above right of=3] {\small{DGU $4$}};
  \node[main node] (5) [below right of=4] {\small{DGU $5$}};
  \node[main node] (6) [below of=3] {\small{DGU $6$}};
 
  \path[every node/.style={font=\sffamily\small}]

  (1) edge node [left] {} (2)
  (2) edge node [right] {} (1)
  
  (3) edge node [left] {} (4)
  (4) edge node [right] {} (3)
  
  (2) edge node [left] {} (4)
  (4) edge node [right] {} (2)
  
  (4) edge node [left] {} (5)
  (5) edge node [right] {} (4);

  \draw[red] (1) to (3);
  \draw[red] (3) to (1);
  \draw[red] (3) to (4);
  \draw[red] (4) to (3);

  \draw[green] (1) to (6);
  \draw[green] (6) to (1);
  \draw[green] (5) to (6);
  \draw[green] (6) to (5);

\end{tikzpicture}
                 \caption{Graph $\GG_{el}$ of the mG
                   composed of DGUs 1-5 until $t = 4$ s and
                   of 6 DGUs after the plugging-in of
                   $\subss{\hat{\Sigma}}{6}^{DGU}$ (in green). At time
                   $t = 12$ s, DGU 3 is removed (in red).}
                 \label{fig:5areasplug}
               \end{figure}
We assume that DGUs 1-5 supply  $10\mbox{ } \Omega$,
$6\mbox{ } \Omega$, $4\mbox{ } \Omega$, $2 \mbox{ }\Omega$ and
$3\mbox{ } \Omega$ resistive loads, respectively. In PnP controllers
$\subss \CC i$, no compensators $\tilde C_{i}$ and $N_i$ have
been used. 

At $t = 0$, DGUs 1-5 are interconnected and equipped with controllers
$\subss{\CC}{i}$, $i = 1,\dots,5$, produced by Algorithm
\ref{alg:ctrl_design}. 

\subsection*{Plug-in of a new DGU}
At time $t = 4$ s, we simulate the connection of DGU
$\subss{\hat{\Sigma}}{6}^{DGU}$ with
$\subss{\hat{\Sigma}}{1}^{DGU}$ and
$\subss{\hat{\Sigma}}{5}^{DGU}$ (as shown in Figure~\ref{fig:5areasplug}) so as to assess the PnP
capabilities of our design procedure. According to
the plug-in protocol described in Section
\ref{sec:PnP}, one must run
Algorithm \ref{alg:ctrl_design} only for designing
$\subss \CC
6$. As the Algorithm does not stop in Step
\ref{enu:stepAalgCtrl}, the plug-in of DGU 6 is 
performed and, most importantly, no update of the controllers $\CC_{[j]}$,
$j\in{\NN_{6}}$, with $\NN_{6}=\{1,5\}$ is required. 
Figure \ref{fig:Sc2_V_0}
illustrates voltages at PCCs 1, 5 and 6 around the plug-in time. Although we have different
voltages at PCCs, we notice
very small deviations of the output signals of DGUs 1, 5
and 6 from their
references when DGU 6 is plugged-in. Moreover, since no
switch of controller is performed, these perturbation are
much smaller than those in the corresponding experiments
shown in \cite{Tucci2015a}.
    \begin{figure}[!htb]
                      \centering
                      \begin{subfigure}[htb]{0.48\textwidth}
                        \centering
                           \includegraphics[width=1\textwidth]{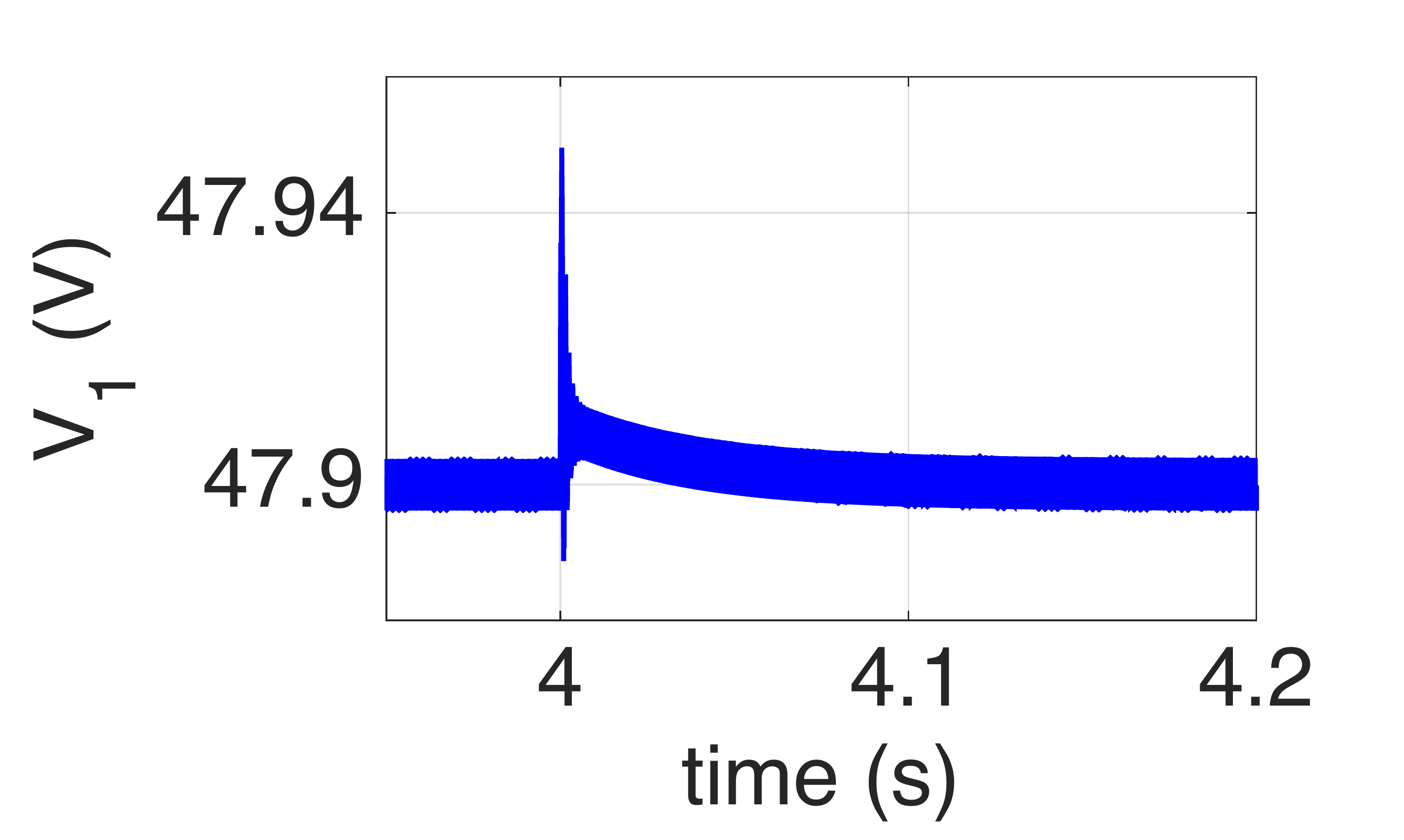}
                        \caption{Voltage at $PCC_1$.}
                        \label{fig:Sc2_V1}
                      \end{subfigure}
                      \begin{subfigure}[htb]{0.48\textwidth} 
                        \centering
                       \includegraphics[width=1\textwidth]{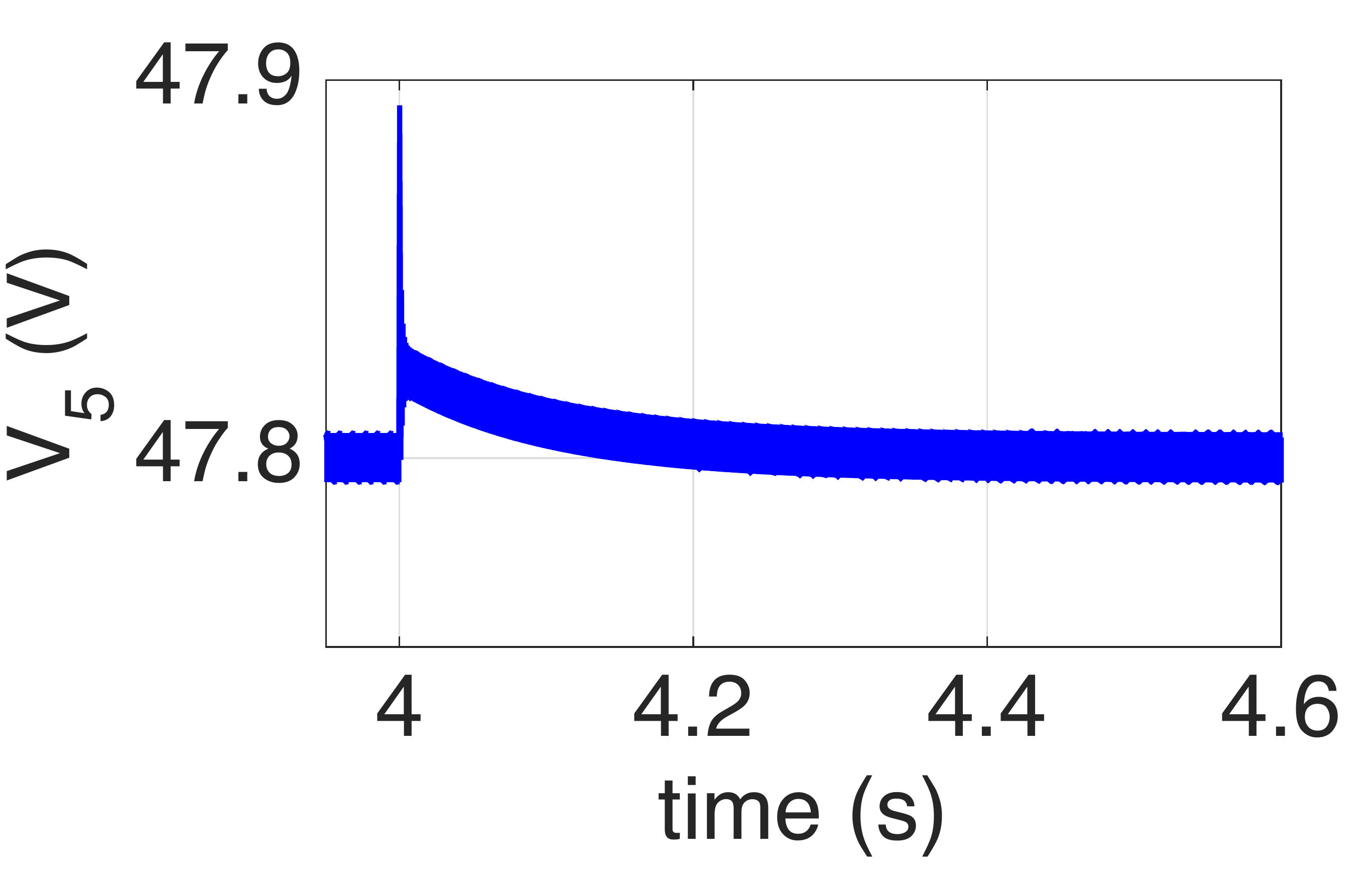}
                        \caption{Voltage at $PCC_5$.}
                        \label{fig:Sc2_V5}
                      \end{subfigure}
                      \begin{subfigure}[htb]{0.48\textwidth}
                        \centering
                        \includegraphics[width=1\textwidth]{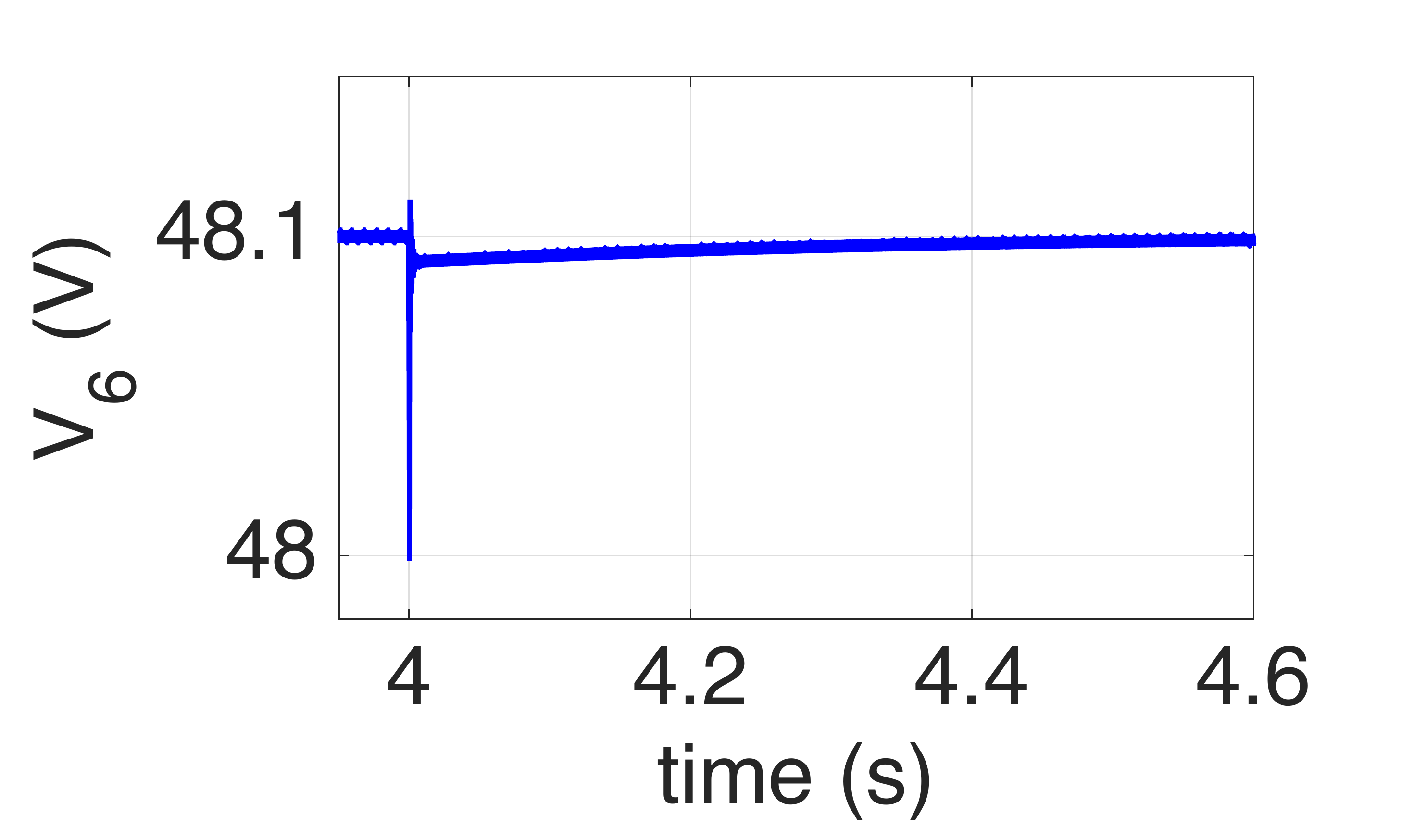}
                        \caption{Voltage at $PCC_6$.}
                        \label{fig:Sc2_V6}
                      \end{subfigure}
                     \caption{Performance of PnP
                       decentralized voltage controllers during the
                       plug-in of DGU 6 at time $t=4$ s.}
                      \label{fig:Sc2_V_0}                 
                    \end{figure}

\subsection*{Robustness to unknown load changes}
At $t=8$ s, the load of DGU 6 is decreased from $8
\mbox{ }\Omega$ to $4
\mbox{ }\Omega$. As shown in Figures
\ref{fig:Sc2_V1}-\ref{fig:Sc2_V5}, 
right after $t=8$ the voltages at $PCC_1$ and
$PCC_5$ slightly deviate from their references. However, these
oscillations disappear after very short transients. A similar behavior
can be noticed for the PCC voltage of DGU 6 (Figure~\ref{fig:Sc2_V6}).
\begin{figure}[!htb]
                \centering
                      \begin{subfigure}[htb]{0.48\textwidth} 
                        \centering    
                        \includegraphics[width=1\textwidth]{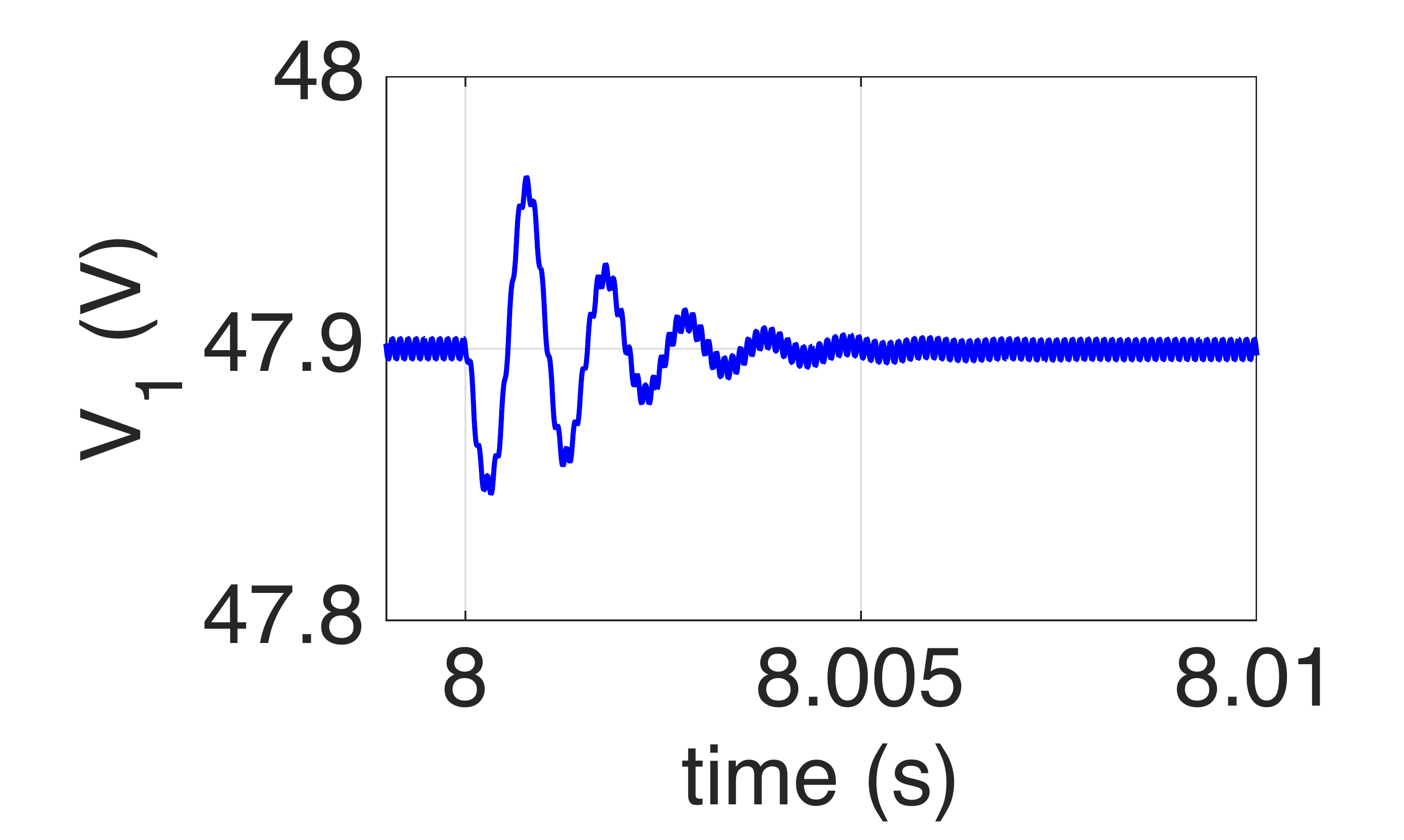}
                        \caption{Voltage at $PCC_1$.}
                        \label{fig:Sc2_V1}
                      \end{subfigure}
                      \begin{subfigure}[htb]{0.48\textwidth}
                       \centering
                       \includegraphics[width=1\textwidth]{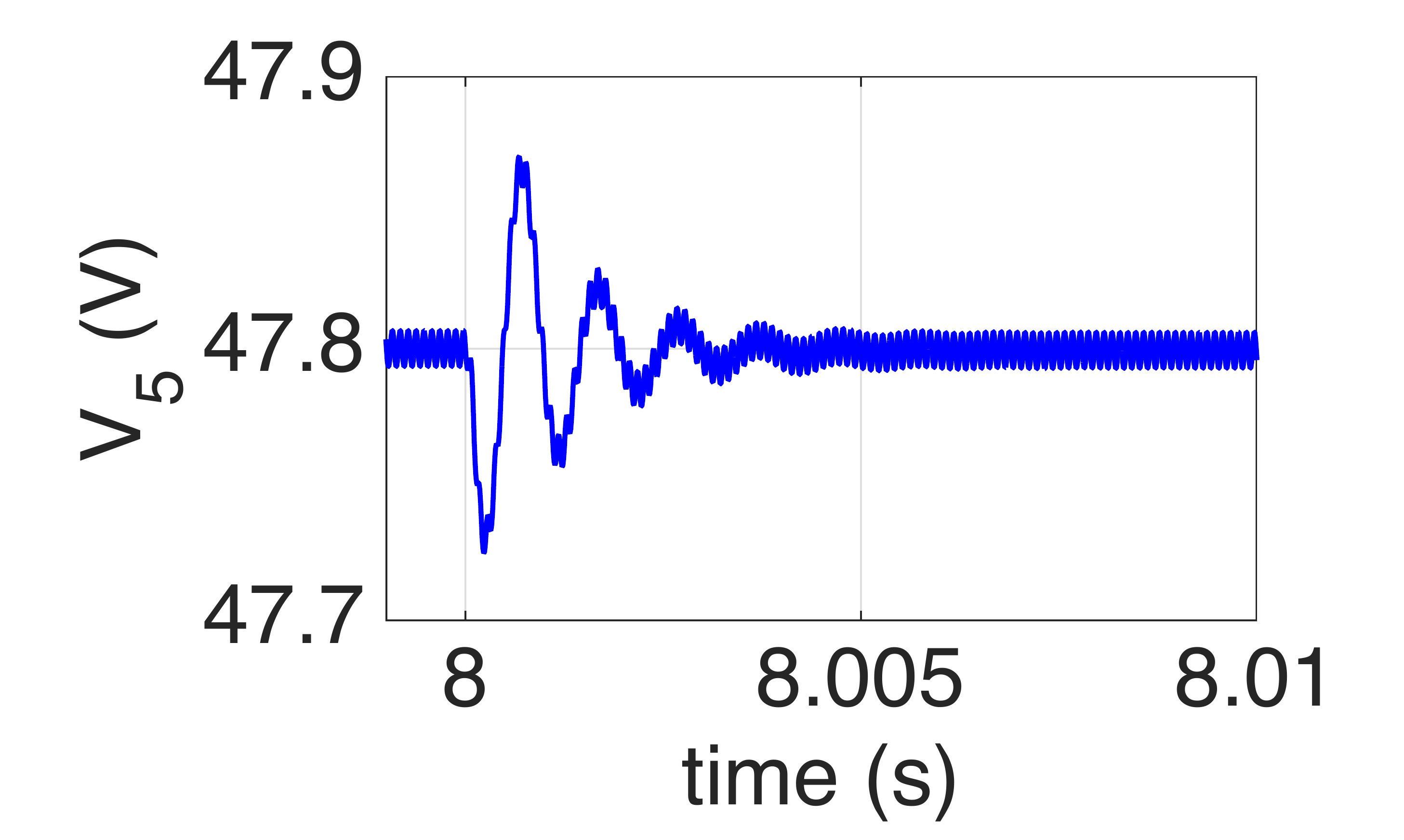}
                        \caption{Voltage at $PCC_5$.}
                        \label{fig:Sc2_V5}
                      \end{subfigure}
                      \begin{subfigure}[htb]{0.48\textwidth}\hspace{2mm} 
                       \centering
                       \includegraphics[width=1\textwidth]{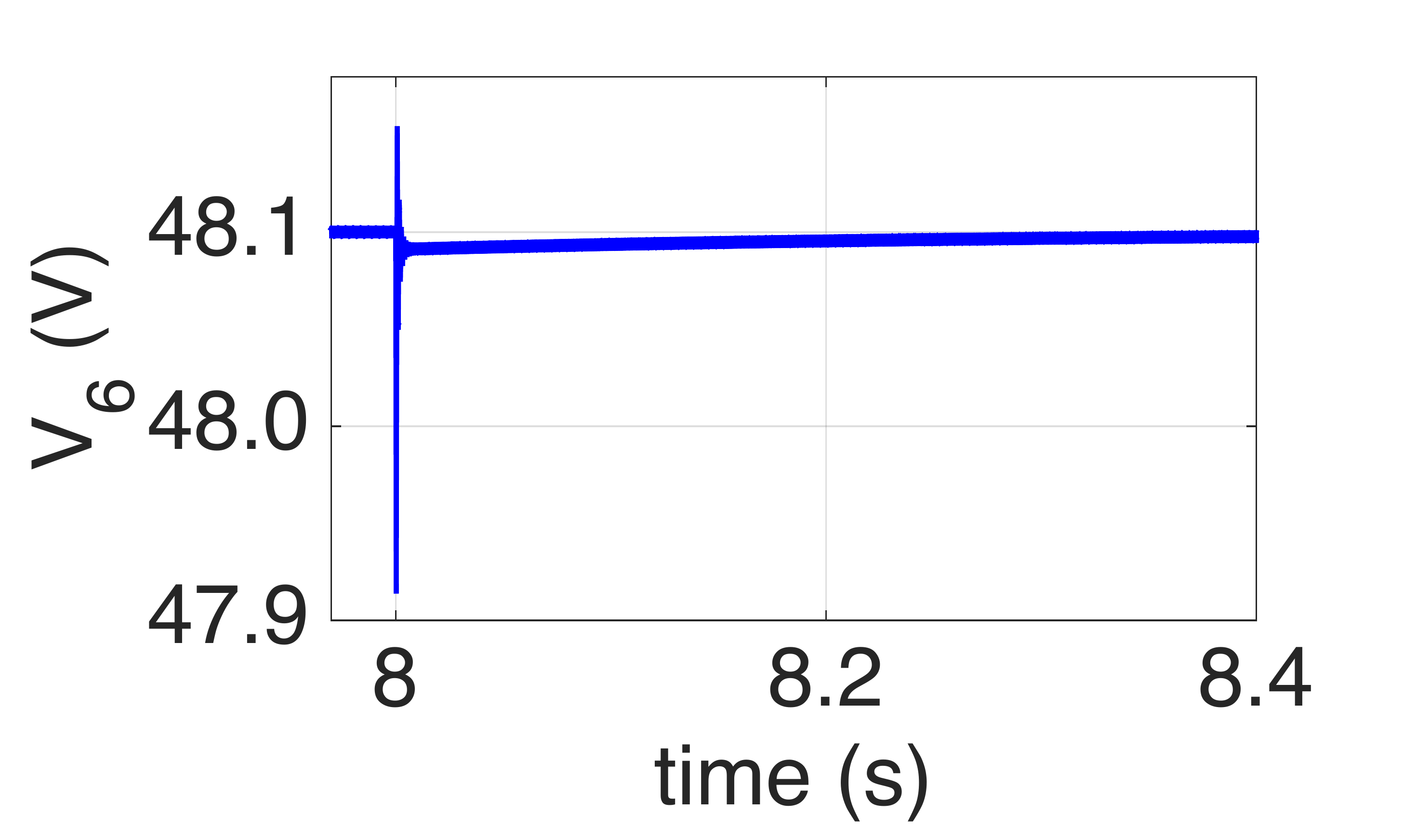}
                        \caption{Voltage at $PCC_6$.}
                        \label{fig:Sc2_V6}
                      \end{subfigure}
                     \caption{Performance of PnP
                       decentralized voltage controllers in terms of
                       robustness to an abrupt change of load
                       resistances at time $t = 8$ s.}
                      \label{fig:Sc2_V_1}                 
                    \end{figure}
\subsection*{Unplugging of a DGU}
At time $t = 12 s$ we perform the disconnection of
$\subss{\hat{\Sigma}}{3}^{DGU}$
(see Figure \ref{fig:5areasplug}). 
As described in Section \ref{sec:PnP},
no controller update is required for the DGUs that were connected to
it (i.e. DGUs 1 and 4). Figure \ref{fig:Sc2_V_2} shows voltages at PCCs 1 and 4 around the
unplugging event. Since controllers
$\CC_{[1]}$ and $\CC_{[4]}$ do not need to be updated,
subsystems $\subss{\hat{\Sigma}}{j}^{DGU}$, $j\in\NN_3$, show deviations
from their respective references which are smaller than those provided
in \cite{Tucci2015a}.
 \begin{figure}[!htb]
	
	\begin{subfigure}[!htb]{0.48\textwidth}
		\vspace{-1.1mm}
		\centering
		\includegraphics[width=1\textwidth]{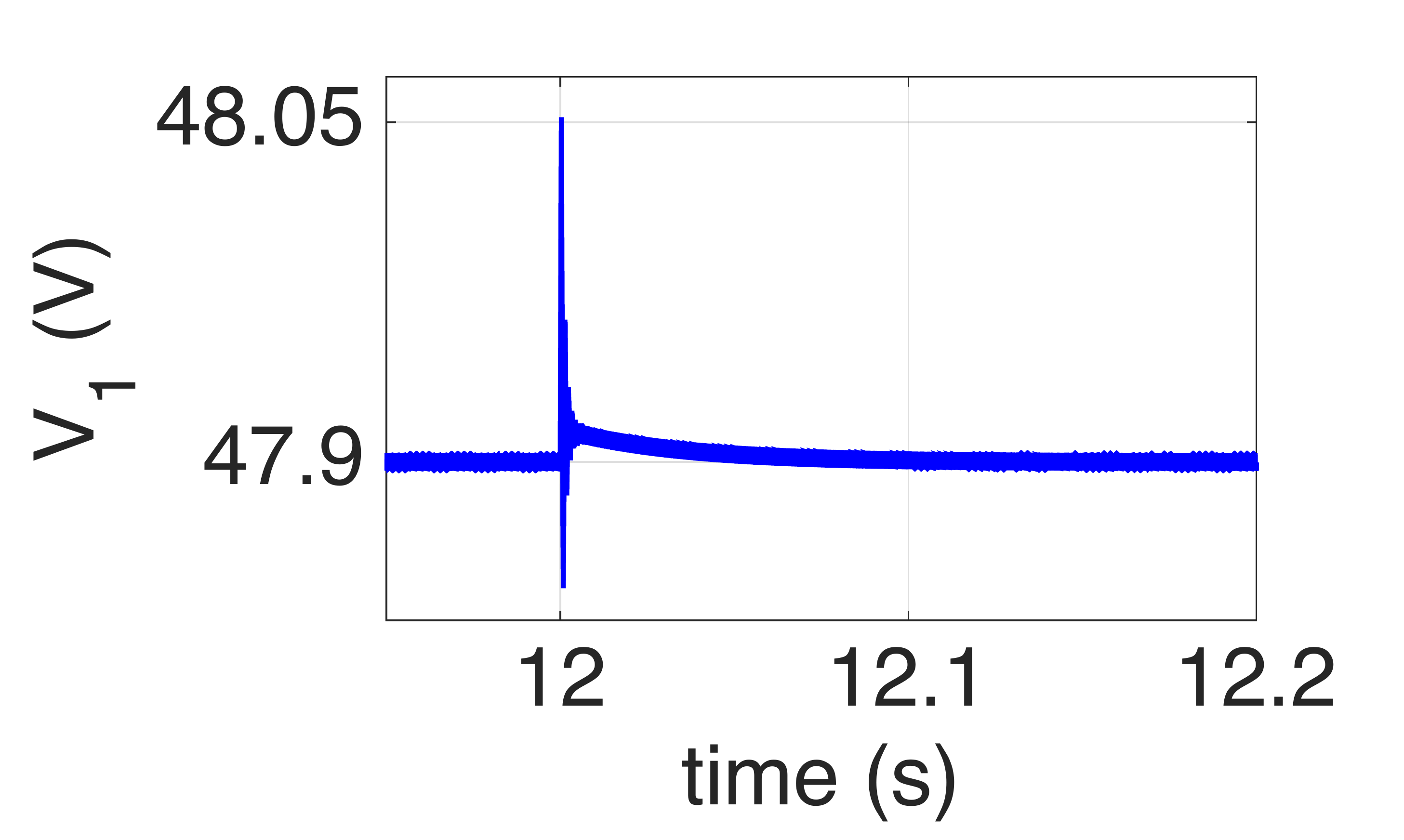}
		\caption{Voltage at $PCC_1$.}
		\label{fig:Sc2_V1_unplugging}
	\end{subfigure}
	\begin{subfigure}[!htb]{0.48\textwidth}
		\centering
		\includegraphics[width=1\textwidth]{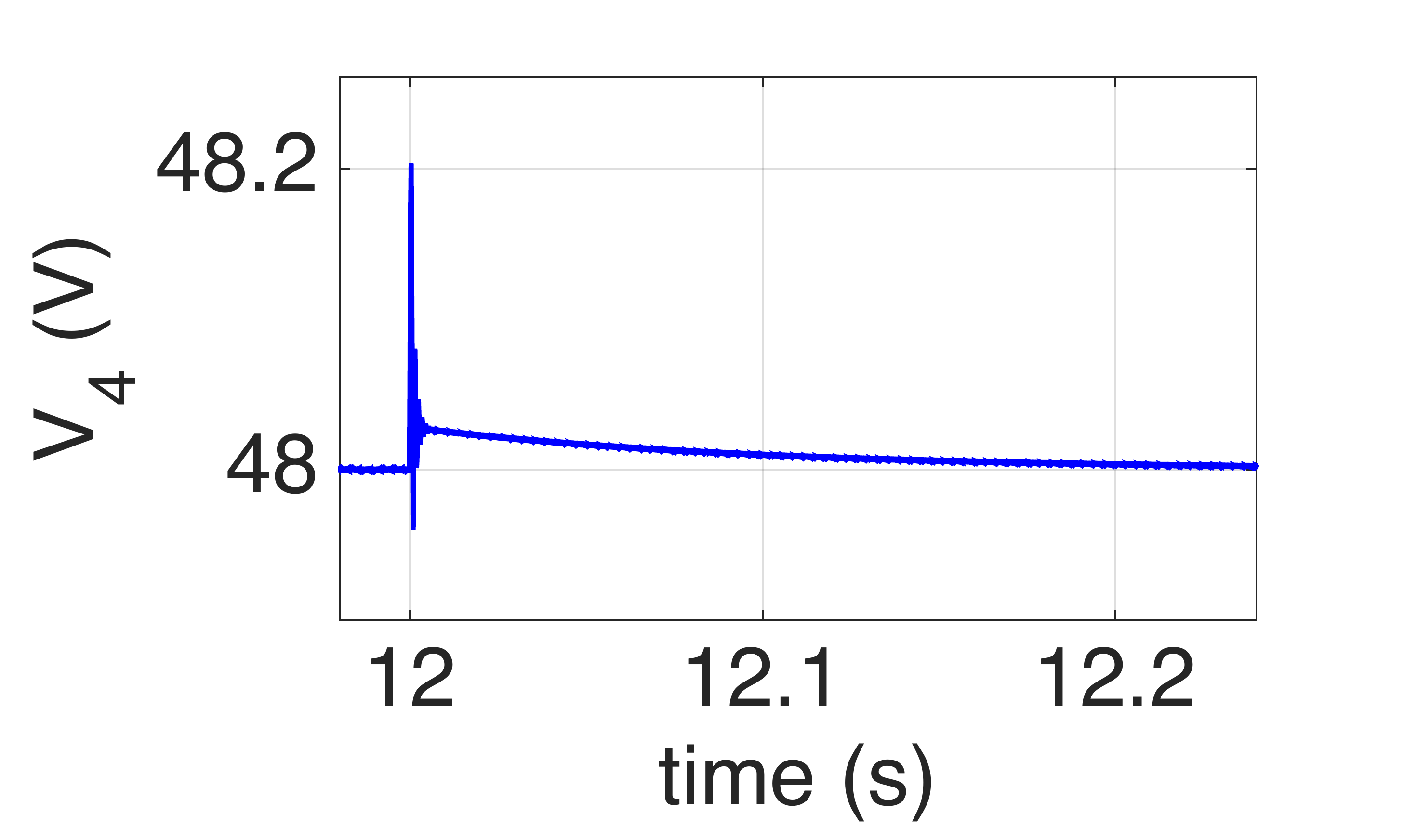}
		\caption{Voltage at $PCC_4$.}
		\label{fig:Sc2_V4_unplugging}
	\end{subfigure}
	\caption{Performance of PnP
		decentralized voltage controllers during the
		unplugging of DGU 3 at $t=12$ s.}
	\label{fig:Sc2_V_2}                 
\end{figure}

\section{Conclusions}
\label{sec:conclusions}
In this paper, a totally decentralized scalable approach for voltage
regulation in DC mG has been presented. Differently from the PnP design
algorithm in \cite{Tucci2015a,Tucci2015dc}, the synthesis of
local controllers does not require knowledge of power line
parameters. Moreover, whenever the plug-in (resp. -out) of a DGU is
required, its future (resp. previous) neighbors do not have to retune
their local controller, thus considerably simplifying the PnP protocol while still
guaranteeing the overall mG stability. Future research will focus on
characterizing how performance of the closed-loop mG, depends on the
topology of the electrical graph. We will also study how to couple the
new PnP primary control layer with higher-level regulation schemes for
achieving, e.g., current sharing.

\clearpage
\appendix
\section{How interactions among DGUs can destabilize the mG}
\label{sec:AppUnstable}
	In this Appendix, we show why designing decentralized stabilizing controllers without taking into account Assumptions \ref{ass:ctrl} and \ref{ass:equal_ratio} (i.e. without counteracting the contribution of coupling terms on the total system energy computation) may lead to mG instability when DGUs are interconnected. 
	
	We consider two DGUs with dynamics 
	\begin{equation}
          \label{eq:subsystemsdec}
          \begin{aligned}
            \subss{\hat{\Sigma}}{1} &: \left\lbrace \begin{aligned}
		\subss{\dot{\hat{x}}}{1}(t) &= \hat{A}_{11}\subss{\hat{x}}{1}(t)+\hat{A}_{12}\subss{\hat{x}}{2}(t)+\hat{B}_1 \subss{u}{1}(t)+\hat{M}_1 \subss{\hat{d}}{1}(t)\\
		\subss{y}{1}(t)       &= \hat{C}_1 \subss{\hat{x}}{1}(t)
              \end{aligned}
            \right.\\
            \subss{\hat{\Sigma}}{2} &: \left\lbrace \begin{aligned}
		\subss{\dot{\hat{x}}}{2}(t) &= \hat{A}_{22}\subss{\hat{x}}{2}(t)+\hat{A}_{21}\subss{\hat{x}}{1}(t)+\hat{B}_2 \subss{u}{2}(t)+\hat{M}_2 \subss{\hat{d}}{2}(t)\\
		\subss{y}{2}(t)       &= \hat{C}_2 \subss{\hat{x}}{2}(t)
              \end{aligned}
            \right.
          \end{aligned}
	\end{equation}
        where $\subss{\hat{x}}{i}=[V_{i},I_{ti},v_{i}]^T$, $\subss{u}{i}=V_{ti}$, $\subss{y}{i}=V_{i}$, $\subss{\hat{d}}{i}=[I_{Li},V_{ref,i}]^T$, $i = 1,2$, are, respectively, the state, the control input, the controlled variable and the exogenous input (see Section \ref{sec:State-space model of the mG}). Matrices in \eqref{eq:subsystemsdec} have the same structure as in \cite{Tucci2015a}, i.e., since in this example $\NN_1 =\{2\}$ and $\NN_2 = \{1\}$, one has
        \begin{equation}
        A_{ii}=\begin{bmatrix}
        -\frac{1}{R_{ij}C_{ti}} & \frac{1}{C_{ti}} \\
        -\frac{1}{L_{ti}} & -\frac{R_{ti}}{L_{ti}} \\
        \end{bmatrix},\hspace{3mm}i = 1, 2.
        \end{equation}
      Electrical parameters, which are similar to those in \cite{shafiee2014modeling}, are reported in Table \ref{tbl:mG_parameters}. 
        
        In the sequel, we separately analyze the impact of couplings on the mG stability when DGUs are locally stabilized via either Linear Quadratic Regulators (LQRs) or through pole placement design.
        \subsection*{Linear Quadratic Regulators}
		We design decentralized controllers for each DGU assuming that they are dynamically decoupled, hence $\hat{A}_{12}=\hat{A}_{21}=0$. Since the state $\subss\hx i$ is measured, we can design the following state-feedback decentralized controllers
	\begin{equation}
          \label{eq:ctrlLaws}
          \begin{aligned}
            \subss{u}{1}(t)&=K_{1}\subss{\hat{x}}{1}(t)\\
            \subss{u}{2}(t)&=K_{2}\subss{\hat{x}}{2}(t)
          \end{aligned}
	\end{equation}
        where $K_{1}$ and $K_{2}$ are LQRs computed using the weights $Q_{1}=\diag(10^{-3}, 10^{-2}, 10^3)$, $R_1=0.1$ and $Q_{2}=\diag(10^{-2},10^{-2},10^4)$, $R_{2}=10^{-2}$, respectively. Control laws \eqref{eq:ctrlLaws} guarantee that the closed-loop decoupled DGUs
	\begin{equation}
          \label{eq:sysdecex}
          \begin{aligned}
          \begin{bmatrix}
            \subss{\dot{\hat{x}}}{1}(t)\\
            \subss{\dot{\hat{x}}}{2}(t)
          \end{bmatrix}&=\begin{bmatrix}
            \hat{A}_{11} & 0\\
            0 & \hat{A}_{22}\\
          \end{bmatrix}
          \begin{bmatrix}
            \subss{\hat{x}}{1}(t)\\
            \subss{\hat{x}}{2}(t)
          \end{bmatrix}
          +\begin{bmatrix}
            \hat{B}_{1}K_{1} & 0\\
            0 & \hat{B}_{2}K_{2}
          \end{bmatrix}
          \begin{bmatrix}
            \subss{\hat{x}}{1}(t)\\
            \subss{\hat{x}}{2}(t)
          \end{bmatrix}=\\
          &=\underbrace{\begin{bmatrix}
          \hat{A}_{11}+\hat{B}_{1}K_{1} & 0\\
          0 & \hat{A}_{22}+\hat{B}_{2}K_{2}\\
          \end{bmatrix}}_{\mbf{\hat{A}_{CL}^D}}\begin{bmatrix}
      \subss{\hat{x}}{1}(t)\\
  \subss{\hat{x}}{2}(t)
\end{bmatrix}
          \end{aligned}
	\end{equation}
        are asymptotically stable. Indeed, $\mathrm{eig} (\mbf{\hat{A}_{CL}^D})=  \mathrm{eig}(\hat{A}_{11}+\hat{B}_{1}K_{1})\text{ }\cup\text{ }\mathrm{eig}(\hat{A}_{22}+\hat{B}_{2}K_{2})$ is the set
        \begin{equation*}
       \{  -9.0629\cdot 10^3,   -0.0143\cdot 10^{3}, -0.1945\cdot 10^{3}\}\cup\{    -9.9717\cdot 10^{3},-0.0486\cdot 10^{3},-0.6064\cdot 10^{3}\}.
        \end{equation*}

        
        Considering coupling terms, the closed-loop system becomes  
	\begin{equation}
          \label{eq:syscouplex}
          \begin{aligned}
          \begin{bmatrix}
            \subss{\dot{\hat{x}}}{1}(t)\\
            \subss{\dot{\hat{x}}}{2}(t)
          \end{bmatrix}&=\begin{bmatrix}
            \hat{A}_{11} & \hat{A}_{12} \\
            \hat{A}_{21}  & \hat{A}_{22}\\
          \end{bmatrix}
          \begin{bmatrix}
            \subss{\hat{x}}{1}(t)\\
            \subss{\hat{x}}{2}(t)
          \end{bmatrix}
          +\begin{bmatrix}
            \hat{B}_{1}K_{1} & 0\\
            0 & \hat{B}_{2}K_{2}
          \end{bmatrix}
          \begin{bmatrix}
            \subss{\hat{x}}{1}(t)\\
            \subss{\hat{x}}{2}(t)
          \end{bmatrix}=\\
          &= \begin{bmatrix}
          \subss{\dot{\hat{x}}}{1}(t)\\
          \subss{\dot{\hat{x}}}{2}(t)
          \end{bmatrix}\underbrace{\begin{bmatrix}
          \hat{A}_{11}+ \hat{B}_{1}K_{1} & \hat{A}_{12} \\
          \hat{A}_{21}  & \hat{A}_{22}+ \hat{B}_{1}K_{1}\\
          \end{bmatrix}}_{\mbf{\hat{A}_{CL}^C}}\begin{bmatrix}
          \subss{\dot{\hat{x}}}{1}(t)\\
          \subss{\dot{\hat{x}}}{2}(t)
          \end{bmatrix}.
          \end{aligned}
	\end{equation}
        Since the controllers have been designed without taking into account interactions among DGUs, we cannot ensure that system \eqref{eq:syscouplex} is asymptotically stable. In fact, for the proposed example, we have $$\text{eig}(\mbf{\hat{A}_{CL}^C}) =\{-19077, \mbf{20 +560\mathrm{i}}, \mbf{20 -560\mathrm{i}},  -690, -161, -11\}.$$
	    \subsection*{Pole placement design} An alternative method for designing decentralized stabilizing controllers \eqref{eq:ctrlLaws}, is to assume again $\hat{A}_{12}=\hat{A}_{21}=0$ and place the closed-loop poles of subsystems $ \subss{\hat{\Sigma}}{1}$ and $ \subss{\hat{\Sigma}}{2}$ (i.e. the eigenvalues of $\mbf{\hat{A}_{CL}^D}$) in the left half plane. In particular, since the pair $(\hat A_{ii}, \hat B_{i})$ is controllable (see Proposition 2 in \cite{Tucci2015a}), we can set 
        \begin{equation}
        \label{eqn:pole_placement}
        \begin{aligned}
        \text{eig}(\hat{A}_{11}+\hat{B}_{1}K_{1}) &= \{-8.5190\cdot 10^3, -530.4, -1.46\}\\
        \text{eig}(\hat{A}_{22}+\hat{B}_{2}K_{2}) &= \{-9.3734\cdot 10^3, -571.9, -1.44\}
        \end{aligned}
        \end{equation}
and derive gains $K_1$ and $K_2$ satisfying \eqref{eqn:pole_placement} by using the algorithm in \cite{kautsky1985robust}. Obviously, the obtained controllers stabilize the closed-loop decoupled subsystems \eqref{eq:sysdecex}. However, they cannot guarantee that the interconnection of DGUs 1 and 2 (i.e. system \eqref{eq:syscouplex}) is asymptotically stable. Indeed, we get
$$ \text{eig}(\mbf{\hat{A}_{CL}^C}) = \{-18803,\mbf{23 + 2319\mathrm{i}},\mbf{23 - 2319\mathrm{i}},-237,  -0.16, -0.13 \}.$$


  \begin{table}[!htb]
	\caption{Electrical parameters of the mG with dynamics \eqref{eq:subsystemsdec}.}	
	\label{tbl:mG_parameters}
	\centering
	\begin{tabular}{*{4}{c}}
		\toprule
		 \multicolumn{4}{c}{\textbf{Converter parameters}} \\
		\midrule
		DGU & $R_t$ $(\Omega)$ & $L_{t}$ (mH) & $C_t$ (mF)\\
		\midrule
		 $\subss{\hat{\Sigma}}{1}$ & 0.1  & 1.8 & 2.2 \\
		 $\subss{\hat{\Sigma}}{2}$& 0.2 & 1.7 & 2 \\
		\midrule
	 \multicolumn{4}{c}{\textbf{Transmission line parameters}}\\
		\midrule
		 \multicolumn{2}{c}{$R_{12}$ $(\Omega)$}&  \multicolumn{2}{c}{$L_{12}$ ($\mu$H)}\\
		 \midrule
		\multicolumn{2}{c}{0.05} & \multicolumn{2}{c}{1.8} \\
		\bottomrule
	\end{tabular}
\end{table}
\clearpage
\section{Feasibility of the LMI test}
\label{sec:LMI_feasibility}
This Appendix summarizes the studies we performed in order to (i) evaluate the applicability of our \textit{line-independent} control design procedure, and (ii) provide a proper comparison between the proposed approach and our previous work \cite{Tucci2015a}. For both the analyses, LMIs have been solved in MatLab/Yalmip, using SeDuMi solver.  
	\subsection{\textit{Line-independent} design: conservativity of the plug-in test \eqref{eq:optproblem}}
	\label{subsec:Line_independent_design}
In order to better assess the applicability of the \textit{line-independent} control design procedure (and to provide a guideline to the choice of $\bar\sigma$ in \eqref{eq:equal_ratio}), we performed the following extensive analysis. For $\bar\sigma = 10$, we solved LMI \eqref{eq:optproblem} considering different values of the DGU's converter parameters $(R_{t}, L_{t}, C_{t})$, and check when LMIs are infeasible. 

Figure \ref{fig:LMI_tests} shows combinations of these parameters for which the LMIs are feasible (blue circles), and infeasible (red stars). Although we used wide ranges for converter's parameters, numerical results reveal that, for values typically found in the literature for low-voltage DC mGs\footnote{See, e.g., \cite{7438847}, \cite{shafiee2014modeling} and \cite{6497633}.} (i.e. the points within the green box in Figure \ref{fig:LMI_tests}), LMIs are always feasible. This confirms the validity of the proposed control design methodology. 
\begin{figure}[!htb]
	\centering
	\includegraphics[scale=0.3]{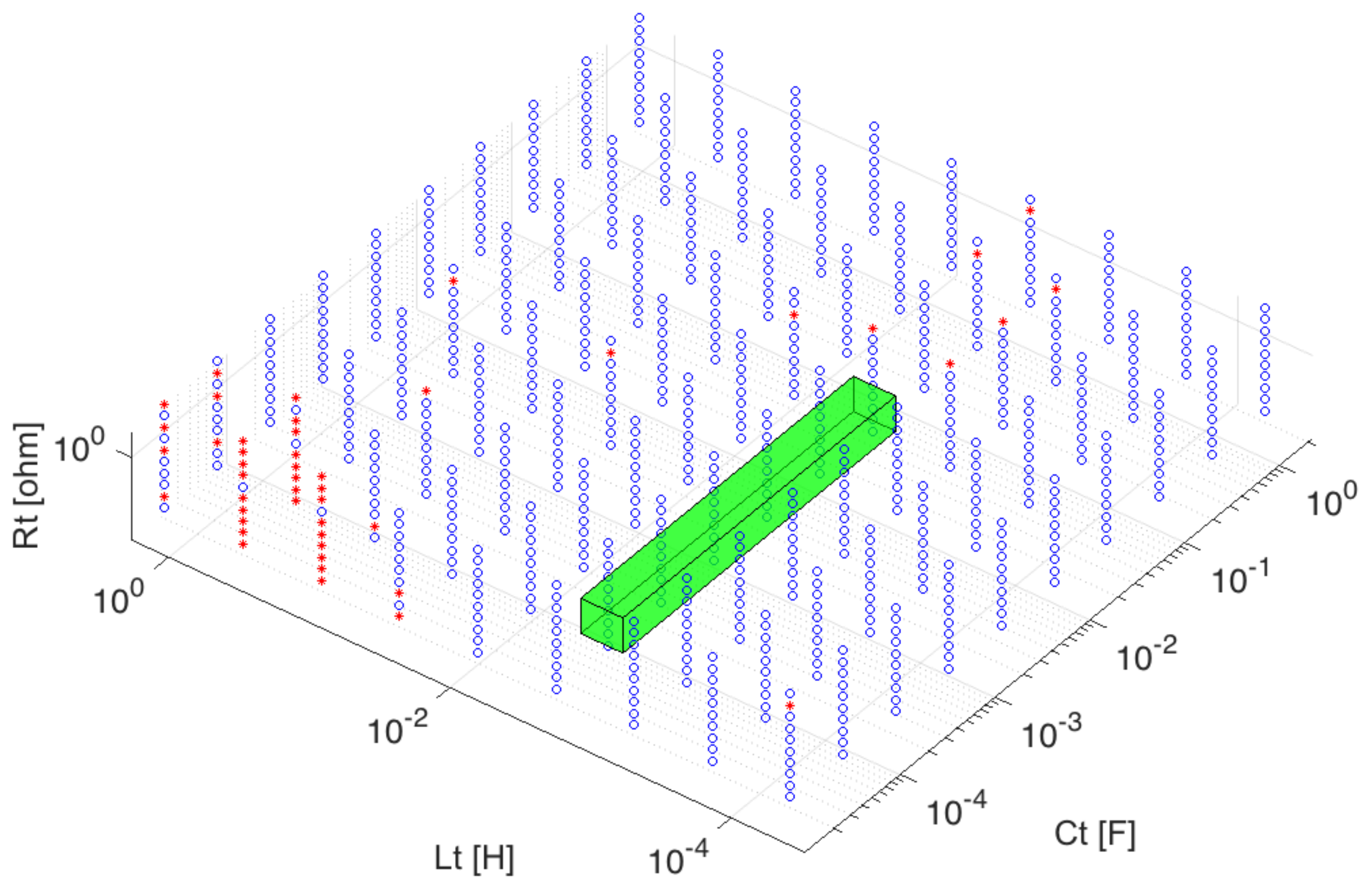}
	\caption{LMI results for combinations of $R_{t}$, $L_{t}$ and $C_{t}$. Blue circles indicate feasible LMIs while red stars correspond to infeasible ones. The green box encloses typical DGU parameters for low-voltage DC mGs.}
	\label{fig:LMI_tests}
\end{figure}

	\subsection{Comparison with the \textit{line-dependent} approach in \cite{Tucci2015a}}
	In order to compare the proposed design methodology with the one in \cite{Tucci2015a}, we show the existence of a limit on the maximum number of subsystems which can be connected to the PCC of a given DGU (say DGU $i$), before obtaining the infeasibility of the $i$-th local \textit{line-dependent} plug-in test (i.e. the LMI (25) in \cite{Tucci2015a}). 
	
	We start by considering the interconnection of DGUs 1 and 2 at stage $k=1$ (see Figure \ref{fig:max_num_DGUs_1}). Then, at each stage $k>1$, we
	solve the LMI (25) in \cite{Tucci2015a} for DGU 1 (using $\eta = 10^{-4}$), connecting one DGU at a time to $PCC_1$, thus creating the star topology shown in Figure \ref{fig:max_num_DGUs_N}. Electrical and optimization parameters used for this analysis are reported in Table \ref{tbl:el_parameters}. 
	
	Numerical results reveal that the feasibility test for DGU 1 fails when the plug-in of DGU 4 is requested. This is due to the fact that the design of each local controller depends on the parameters of power lines connecting the corresponding unit to its neighbors.
	On the other hand, in Appendix \ref{subsec:Line_independent_design} we have shown that, for $\bar\sigma=0$, the \textit{line-independent} LMIs \eqref{eq:optproblem} are always feasible.

 \begin{figure}[!htb]
	\centering
	\begin{subfigure}[b]{0.48\textwidth}
		\centering
		\includegraphics[width=.7\textwidth]{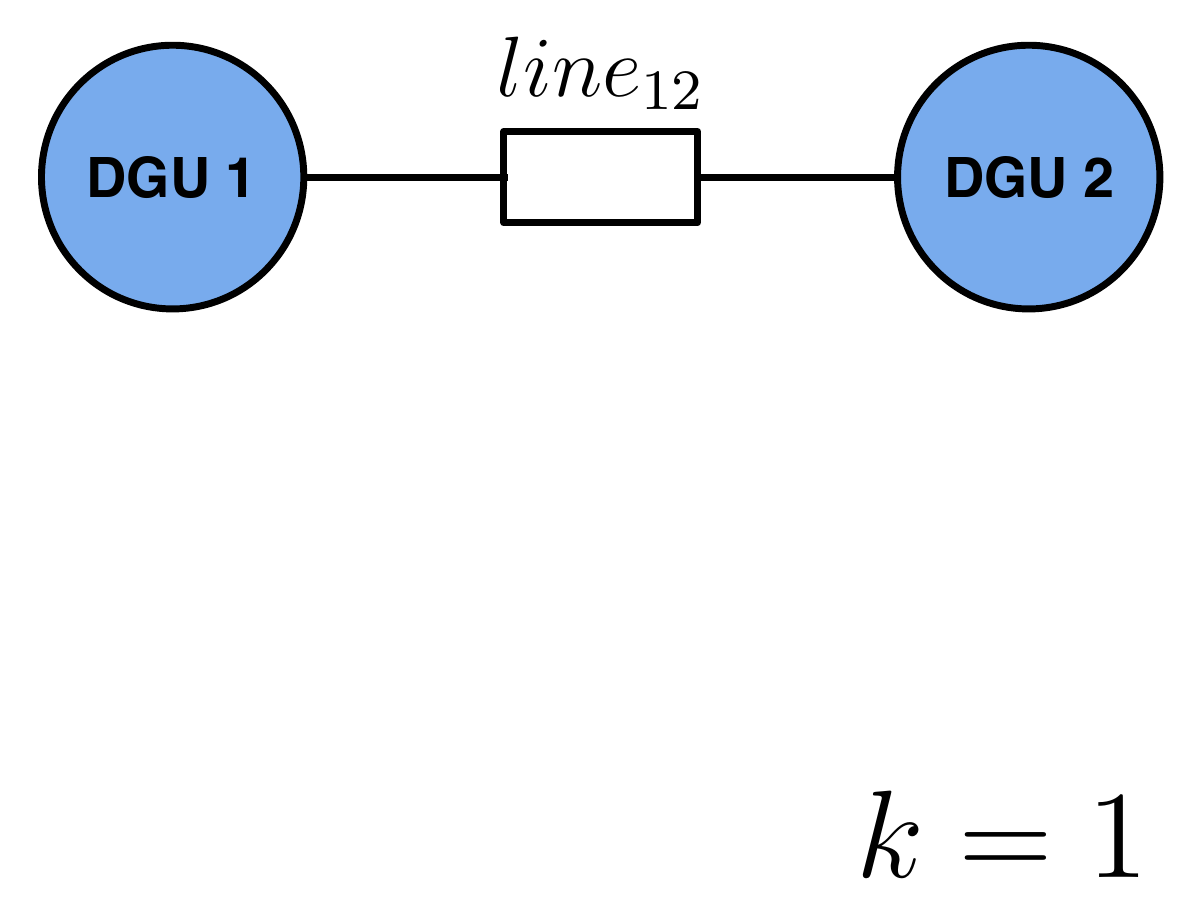}
		\caption{Starting mG.}
		\label{fig:max_num_DGUs_1}
	\end{subfigure}
	\begin{subfigure}[b]{0.48\textwidth} 
		\centering
		\includegraphics[width=1\textwidth]{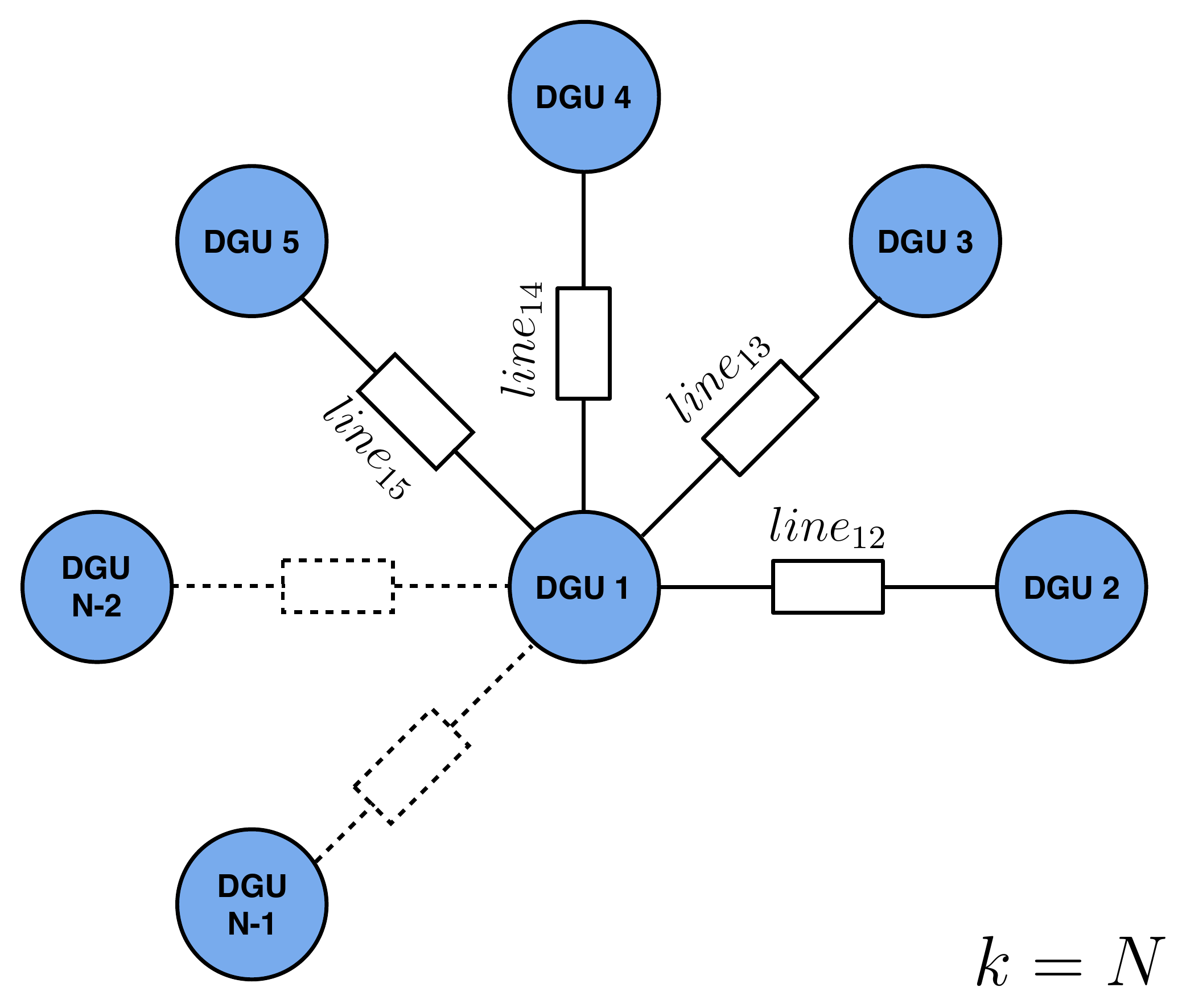}
		\caption{Star-connected mG after $N$ stages.}
		\label{fig:max_num_DGUs_N}
	\end{subfigure}
\caption{Finding the maximum number of DGUs which can be connected $PCC_1$ before obtaining a plug-in request failure.}
\end{figure}
  \begin{table}[!htb]
	\caption{Electrical and optimization parameters.}	
	\label{tbl:el_parameters}
	\centering
	\begin{tabular}{*{4}{c}}
		\toprule
		 \multicolumn{4}{c}{\textbf{Converter parameters}} \\
		\midrule
		DGU & $R_t$ $(\Omega)$ & $L_{t}$ (mH) & $C_t$ (mF)\\
		\midrule
		 $\subss{\hat{\Sigma}}{1}$ & 0.2  & 1.8 & 2.2 \\
		 $\subss{\hat{\Sigma}}{2}$& 0.3 & 2 & 2.2 \\
		 $\subss{\hat{\Sigma}}{3}$& 0.1 & 2.2 & 2.2 \\
		 $\subss{\hat{\Sigma}}{4}$& 0.5 & 3 & 2.2 \\
		\midrule
		\multicolumn{4}{c}{\textbf{Transmission line parameters}}\\
		  \toprule
	\multicolumn{2}{c}{Connected DGUs $(i,j)$} & Resistance $R_{ij} (\Omega)$ & Inductance
		$L_{ij} (\mu H)$ \\
		 \midrule
		\multicolumn{2}{c}{$(1,2)$} & 0.05 & 2.1 \\
		\multicolumn{2}{c}{$(1,3)$} & 0.07 & 1.8 \\
		\multicolumn{2}{c}{$(1,4)$} & 0.03 & 2.5 \\
		\midrule
	 \multicolumn{4}{c}{\textbf{Optimization parameters}}\\
		\midrule
		 \multicolumn{1}{c}{$\alpha_{1}$}&  \multicolumn{2}{c}{$\alpha_{2}$}&
		  \multicolumn{1}{c}{$\alpha_{3}$}\\
		 \midrule
		\multicolumn{1}{c}{$10^{-6}$} & 
		\multicolumn{2}{c}{$10^{-2}$} & 
		\multicolumn{1}{c}{$10^{-3}$} \\
		\bottomrule
	\end{tabular}
\end{table}
\clearpage
\bibliographystyle{IEEEtran}
\bibliography{IEEEabrv,PnP_improved}

\end{document}